\documentclass[format=manuscript,natbib=true,authorversion=true,review=false,screen=true]{acmart}

\pdfoutput=1

\usepackage{blindtext}

\usepackage{xpatch}

\makeatletter
\xpatchcmd{\ps@firstpagestyle}{Manuscript submitted to ACM}{}{\typeout{First patch succeeded}}{\typeout{first patch failed}}
\xpatchcmd{\ps@standardpagestyle}{Manuscript submitted to ACM}{}{\typeout{Second patch succeeded}}{\typeout{Second patch failed}}    \@ACM@manuscriptfalse
\makeatother

\renewcommand\footnotetextcopyrightpermission[1]{}
\setcopyright{none}

\newgeometry{centering,marginratio=1:1,vcentering}

\usepackage{booktabs}

\usepackage[ruled]{algorithm2e} 

\SetAlFnt{\small}
\SetAlCapFnt{\small}
\SetAlCapNameFnt{\small}
\SetAlCapHSkip{0pt}
\IncMargin{-\parindent}

\usepackage{proof-dashed}
\usepackage{stmaryrd}
\usepackage{wrapfig}
\usepackage{xspace}
\usepackage{todonotes}
\usepackage{bm}

\usepackage{breakurl}

\usepackage{cleveref}
\usepackage{array,amsmath,arydshln,colortbl}
\definecolor{light-gray}{gray}{0.8}

\newtheorem{remark}{Remark}

\newcommand{\ghl}[1]{#1}
\newcommand{\phl}[1]{#1}
\newcommand{\thl}[1]{#1}

\newcommand{\nohl}[1]{#1}

\newcommand{\substone}[2]{\{#1/#2\}}
\newcommand{\substtwo}[4]{\{#1/#2, #3/#4\}}

\newcommand{\m}[1]{\mathsf{#1}}

\newcommand{\til}[1]{\overline{#1}}
\newcommand{\prf}[1]{\mathcal{#1}}

\newcommand{\bang}{\mathord{!}}

\newcommand{\query}{\mathord{?}}

\newcommand{\sendbase}[3]{#1.(#2 \pp #3)}
\newcommand{\closebase}[1]{#1}
\newcommand{\srvbase}[2]{\bang #1.#2}
\newcommand{\recvtypebase}[2]{#1.#2}
\newcommand{\Casebase}[2]{\nohl{{#1}.\m{case}(#2)}}
\newcommand{\wait}[2]{\nohl{#1().#2}}

\newcommand{\close}[1]{\nohl{\closebase{{#1}[]}}}

\newcommand{\send}[4]{\nohl{\sendbase{{#1}[#2]}{#3}{#4}}}

\newcommand{\recv}[3]{\nohl{{#1}(#2).#3}}
\newcommand{\recvs}[2]{\nohl{{#1}(#2)}}

\newcommand{\piCase}[3]{\nohl{\Casebase{#1}{#2,#3}}}

\newcommand{\emptyCase}[1]{\nohl{\Casebase{#1}{}}}

\newcommand{\inl}[2]{\nohl{{#1}[\m{inl}].#2}}

\newcommand{\inr}[2]{\nohl{{#1}[\m{inr}].#2}}

\newcommand{\sendtype}[3]{\nohl{{#1}[#2].#3}}

\newcommand{\recvtype}[3]{\nohl{\recvtypebase{{#1}(#2)}{#3}}}

\newcommand{\srv}[3]{\nohl{\srvbase{#1(#2)}{#3}}}

\newcommand{\client}[3]{\nohl{\query #1[#2].#3}}

\newcommand{\pp}{\nohl{\ \;\boldsymbol{|}\ \;}}

\newcommand{\res}[1]{(\boldsymbol\nu #1)\,}

\newcommand{\seq}{\vdash}
\newcommand{\gseq}{\vDash}

\newcommand{\dual}[1]{#1^{\bot}}

\newcommand{\tensor}{\otimes}
\newcommand{\tensort}[1]{\otimes^{\phl{#1}}}

\newcommand{\parr}{\mathbin{\bindnasrepma}}
\newcommand{\parrt}[1]{\parr^{\phl{#1}}}
\newcommand{\with}{\mathbin{\binampersand}}
\newcommand{\witht}[1]{{\binampersand}^{\phl{#1}}}
\newcommand{\one}{1}

\newcommand{\zero}{0}

\newcommand{\gentensor}[3]{#1\tensort{#2} #3}
\newcommand{\genparr}[3]{#1\parrt{#2} #3}
\newcommand{\genoplus}[3]{#1\oplust{#2} #3}
\newcommand{\genwith}[3]{#1\witht{#2} #3}

\newcommand{\oplust}[1]{\oplus^{\phl{#1}}}

\newcommand{\reducesto}{\longrightarrow}

\newcommand{\ba}{\begin{array}}
\newcommand{\ea}{\end{array}}
\newcommand{\rname}[1]{\ensuremath{\textsc{#1}}}

\newcommand{\fn}{\m{fn}}
\newcommand{\defeq}{\stackrel{\mbox{\scriptsize def}}{=}}

\newenvironment{equations}{\[\ba{@{}r@{~}c@{~}l@{}}}{\ea\]\ignorespacesafterend}

\newenvironment{equationsl}{\[\ba{@{}r@{~}c@{~}l@{}l@{}}}{\ea\]\ignorespacesafterend}

\newcommand{\gto}{\to}
\newcommand{\gfromto}[2]{#1 \gto #2}

\newcommand{\parrtensor}[4]{\gfromto{#1}{#2}(#3).#4}
\newcommand{\botone}[2]{\closebase{\gfromto{#1}{#2}}}
\newcommand{\withplus}[4]{\piCase{\gfromto{#1}{#2}}{#3}{#4}}
\newcommand{\topzero}[2]{\emptyCase{\gfromto{#1}{#2}}}
\newcommand{\bangquery}[3]{\bang \gfromto{#1}{#2}(#3)}
\newcommand{\forallexists}[4]{\recvtypebase{\gfromto{#2}{#3}}{(#1)#4}}
\newcommand{\provideassume}[3]{\gfromto{#1}{#2}(#3)}
\newcommand{\cpaxiom}{\localaxiom}
\newcommand{\localaxiom}[4]{\nohl{{#1} \gto {#3}^{#4}}}
\newcommand{\globalaxiom}[4]{\nohl{{#1}^{#2} \gto {#3}}}

\newcommand{\cpres}[4]{\nohl{\res{{#1}^{#2}{#3}}}}

\newcommand{\gcpres}[3]{\nohl{\res{#1 : #2}(#3)}}

\newcommand{\CP}{CP\xspace}

\newcommand{\CHOP}{CHOP\xspace}

\newcommand{\HOpi}{HO$\pi$\xspace}
\newcommand{\CLL}{CLL\xspace}

\newcommand{\assume}[1]{\ensuremath{\langle #1 \rangle}}
\newcommand{\tassume}[1]
{\langle #1 \rangle}
\newcommand{\provide}[1]{[#1]}
\newcommand{\sendho}[2]{#1[#2]}
\newcommand{\letp}[2]{[#1 \texttt{:=}\, #2]}
\newcommand{\judge}[3]{#1 \seq #2 :: #3}

\newcommand{\vbold}[1]{\bm{#1}}
\newcommand{\sendhowc}[3]{#1 \vbold{[} #2 \vbold{]} . #3}
\newcommand{\recvhowc}[3]{#1 \vbold{(} #2 \vbold{)} . #3}

\newcommand{\invoke}[2]{#1 \langle #2 \rangle}
\newcommand{\sendfrees}[2]{#1 \langle #2 \rangle}
\newcommand{\sendfree}[3]{#1 \langle #2 \rangle.#3}
\newcommand{\abstr}[2]{(#1)#2}

\renewcommand{\cite}{\citep}

\begin{document}

\title[Classical Higher-Order Processes]{Classical Higher-Order Processes}
\subtitle{(Technical Report)}

\author{Fabrizio Montesi}
\orcid{0000-0003-4666-901X}
\affiliation{%
  \institution{University of Southern Denmark}
  \department{Department of Mathematics and Computer Science}
  \streetaddress{Campusvej 55}
  \city{Odense M}
  \postcode{5230}
  \country{Denmark}}
\email{fmontesi@imada.sdu.dk}

\begin{abstract}
Classical Processes (\CP) is a calculus where the proof theory of classical linear logic types communicating processes with mobile channels, \`a la $\pi$-calculus. Its construction builds on a recent propositions as types correspondence between session types and propositions in linear logic.
Desirable properties such as type preservation under reductions and progress come for free from the metatheory of linear logic.

We contribute to this research line by extending \CP with code mobility. We generalise classical linear logic to capture higher-order (linear) reasoning on proofs, which yields a logical reconstruction of (a variant of) the Higher-Order $\pi$-calculus (\HOpi).
The resulting calculus is called Classical Higher-Order Processes (\CHOP).
We explore the metatheory of \CHOP, proving that its semantics enjoys type
preservation and progress (terms do not get stuck). We also illustrate the expressivity of \CHOP
through examples, derivable syntax sugar, and an extension to multiparty sessions. Lastly, we
define a
translation from \CHOP to \CP, which encodes mobility of process code into reference passing.
\end{abstract}



\thanks{This work was partly supported by Independent Research Fund Denmark, grant no.\
DFF--4005-00304, DFF--1323-00247 and DFF--7014-00041, and by the Open Data Framework project at the
University of Southern Denmark.}

\maketitle

\section{Introduction}
\label{sec:introduction}

Session types \citep{HVK98} define protocols that discipline communications among concurrent
processes, typically given in terms of (variants of) the $\pi$-calculus \citep{MPW92}.
In their seminal paper, \citet{CP10} established a Curry-Howard correspondence between the type
theory of sessions and intuitionistic linear logic, where processes correspond to proofs, session
types to propositions, and communication to cut elimination.
Important properties that are usually proven with additional machinery on top of session types,
like progress, come for free from the
properties of linear logic, like cut elimination.
\citet{W14} revisited the correspondence for Classical Linear
Logic (\CLL), developing the calculus of Classical Processes (\CP). \CP enjoys a stricter
correspondence with linear logic, and branched off a new line of development wrt the
$\pi$-calculus: \CP foregoes the standard semantics of $\pi$-calculus, and instead adopts a semantics that is
dictacted by standard proof transformations found in linear logic.
Therefore, \CP serves two important purposes: it shows how to use linear logic ``directly'' as a
solid foundation for protocol-driven concurrent programming; and it is a convenient model to
explore extensions of linear logic that are interesting in the setting of process models.

In this article, we extend this research line to capture code mobility---the capability to
communicate processes---by logically reconstructing the generalisation of the $\pi$-calculus into
the Higher-Order $\pi$-calculus (\HOpi) \citep{S93}.
Mobile code is widespread. For example, it is a cornerstone of mobile application stores and cloud
computing.
Disciplined programming in these settings is key, since a host may have to run arbitrary code
received from a third party, and such code does might then interact with the host itself and its
environment. This motivates the search for solid foundations for disciplining the programming of
mobile code, as we may hope to do with linear logic.

\paragraph{Outline of development and contributions}
We develop the calculus of Classical Higher-Order Processes (\CHOP), a strict generalisation of \CP that captures process
mobility (\CP is a fragment of \CHOP). \CHOP extends linear logic to higher-order reasoning, by allowing proofs to have
``holes'': a proof may just assume a sequent to be provable and use it as a premise.
Actual proofs that resolve these assumptions can then be provided separately, using a higher-order
version of the cut rule. Providing a proof and assuming provability are respectively typed by two
new kinds of types, $\provide{\Gamma}$ and $\assume{\Gamma}$, where
$\Gamma$ is a typing context.
The rules that we introduce for our new ingredients type the key primitives that elevate \CP to a
higher-order process calculus: send a process, receive a process, store a process in a
variable, and run the process stored in a variable.

Differently from previous higher-order process calculi, \CHOP treats process variables linearly
(following the intuition that it is a higher-order linear logic): a process sent over a channel is
guaranteed to be eventually used exactly once---the receiver may delegate this responsibility
further on to another process. Higher-order types can be composed using the
standard connectives of linear logic, so that we can apply the usual expressivity
of linear logic propositions to this new setting---e.g., multiple assumptions $\assume{\Gamma}
\tensor \assume{\Gamma'}$ or explicit weakening $\query{\assume{\Gamma}}$.

Our higher-order cuts (called chops) support type-preserving proof transformations that yield a
semantics of communication and explicit substitutions (as originally developed by \citet{ACCL91} for the $\lambda$-calculus) for higher-order processes.
Our semantics suggests how the substitutions that should be performed at the receiver as a result
of code mobility can be implemented efficiently. We show that the existing semantics for
link terms in \CP---which bridge two channels and enable behavioural
polymorphism---becomes unsound in the higher-order setting, and then formulate a sound
semantics by generalising $\eta$-expansion in \CP to higher-order I/O. That $\eta$-expansion works as expected
for our new higher-order types also provides evidence of the correct formulation of our new rules.

We investigate the metatheory of \CHOP, showing that our semantics
preserves types and that  processes enjoy progress (communications between communicating processes 
never get stuck).

The design of \CHOP is intentionally minimalistic, since it is intended as a foundational calculus. We illustrate more sophisticated features as extensions that build on top of the basic theory of \CHOP. Specifically, we discuss how to implement higher-order I/O primitives with continuations and multiparty sessions.

We conclude our development by presenting a translation of \CHOP into \CP, which preserves types (up to encoding of higher-order types into channel types of \CP). The translation works by using channel mobility to simulate code mobility, recalling that presented by \citet{S93}.

\paragraph{Structure of the paper}
\Cref{sec:preview} provides an introduction to (the multiplicative fragment of) \CP and an overview of the key insights behind its generalisation to \CHOP.
\Cref{sec:chop} presents \CHOP and a running example that covers its syntax, typing, and semantics.
In \cref{sec:meta}, we study the metatheory of \CHOP.
In \cref{sec:extensions}, we illustrate how to extend \CHOP with syntax sugar and the support for
multiparty sessions.
The translation of \CHOP to \CP is presented in \cref{sec:translation}. We discuss related work in
\cref{sec:related} and present conclusions in \cref{sec:conclusions}.

\paragraph{Differences from previously published material}
A short conference paper describing preliminary work in progress on the development of \CHOP has been previously published~\cite{M17}. That paper presented only a draft of the typing rules for process mobility, which have since been changed (although they are similar, the old ones do not support the
properties that we present here). Aside from motivation and general idea, the present report is thus almost entirely new. Specifically, it includes the following improvements: a more detailed presentation in sections \ref{sec:introduction} and \ref{sec:conclusions}; full definitions (in \cref{sec:chop});
all examples are new (e.g., in \cref{sec:chop}); \cref{sec:meta}, which reports on properties and
proofs, is new; \cref{sec:extensions}, covering extensions, is new; \cref{sec:translation}, the translation from \CHOP to \CP, is new.

\section{Overview: From \CP to \CHOP}
\label{sec:preview}

We give an overview of some basic notions of \CP that we are going to use in this work, and of the key insights behind our technical development.

\subsection{Multiplicative \CP}

First, we recall the multiplicative fragment of the calculus of Classical Processes (\CP), according to its latest version by \citet{CLMSW16}.

\paragraph{Processes and Typing}
In \CP, processes communicate by using channels, which represent endpoints of sessions.
Let $x,y,\ldots$ range over channel names (also called channels, for brevity). Then, the syntax of
processes ($P,Q,R,\ldots$) is the following.
\begin{displaymath}
\begin{array}{rll}
P,Q,R ::=
& \send xyPQ & \text{output a channel} \\
& \recv xyP & \text{input a channel} \\
& \close x & \text{close} \\
& \wait xP & \text{wait} \\
& \cpres{x}{A}{y}{}(P \pp Q) & \text{parallel composition}
\end{array}
\end{displaymath}
An output term $\send xyPQ$ sends a fresh name $y$ over $x$, and then proceeds as the parallel
composition of $P$ and $Q$. Dually, an input term $\recv xyP$ receives a name $y$ on $x$ and then
proceeds as $P$. Names that are sent or received are bound ($y$ in our syntax for output and
input). Term
$\close x$ closes endpoint $x$, and term $\wait xP$ waits for the closing of $x$ and then proceeds
as $P$.
In general, in \CP, square brackets indicate output and round brackets indicate input.
The parallel composition term $\cpres{x}{A}{y}{}(P \pp Q)$ connects the endpoints $x$ at process
$P$ and
$y$ at process $Q$ to form a session, enabling the two processes to communicate. The names $x$ and
$y$ are
bound. The type $A$ describes the behaviour that the processes will follow when communicating.

The types ($A,B,C,\ldots$) of multiplicative \CP are exactly the propositions of multiplicative
classical
linear logic~\cite{G87}. A type abstracts how a channel is used by a process. We recall them below,
together with a description of their interpretation as behavioural types for channels in \CP.
\begin{displaymath}
\begin{array}{rll}
A,B,C ::=
& \gentensor{A}{}{B} & \text{send $A$, proceed as $B$} \\
& \genparr{A}{}{B} & \text{receive $A$, proceed as $B$} \\
& \one & \text{unit for $\tensor$} \\
& \bot & \text{unit for $\parr$}
\end{array}
\end{displaymath}
The notion of duality of linear logic is used in \CP to establish whether two types are compatible.
Formally, we write $\dual A$ for the dual of type $A$. Duality is defined inductively as follows.
\begin{displaymath}
\begin{array}{rcl}
\dual{(A\tensor B)} & = & \dual A \parr \dual B \\
\dual{(A \parr B)} & = & \dual A \tensor \dual B \\
\dual \one & = & \bot \\
\dual \bot & = & \one
\end{array}
\end{displaymath}

A process may use arbitrarily many channels. Therefore, the types of the channels used by a
process are collected in a typing environment ($\Gamma,\Delta,\Sigma,\ldots$). Formally, a typing
environment associates channel names to types, e.g., $\Gamma = x_1:A_1, \ldots, x_n:A_n$ for some natural number
$n$, where all $x_i$ in $\Gamma$ are distinct. Order in environments is ignored (we silently allow for
exchange). Environments can be combined when they do not share names, written $\Gamma,\Delta$.

Using types (and environments), we can type programs in multiplicative \CP (processes) by using the
typing rules given in \cref{fig:mult_cp_rules}.\footnote{The original presentation of \CP has an
equivalent, but slightly different, syntax for typing judgements: $P \seq \Gamma$ instead of our 
$\judge{\ }{P}{\Gamma}$. Our syntax is 
convenient for the technical development that comes in the next sections, since we will add a separate higher-order context ($\Theta$) on the left: $\judge{\Theta}{P}{\Gamma}$.} In the rules, a typing
judgement $\judge{\ }{P}{x_1:A_1,\ldots,x_n:A_n}$ reads ``process $P$ implements behaviour $A_i$ on
each channel $x_i$''.
\begin{figure}[t]
  \begin{displaymath}
    \begin{array}{c}
      \infer[\tensor]
      {
        \judge{}{\phl{\send xyPQ}}
        {\thl{\Gamma}, \thl{\Delta}, \phl x:\thl{A \tensor B}}
      }
      {
        \judge{}{\phl P}
        {\thl{\Gamma}, \phl y:\thl A}
        &
        \judge{}{\phl Q}
        {\thl \Delta, \phl x: \thl B}
      }
      \qquad
      \infer[\parr]
      {\judge{}{\phl{\recv xyP}}
        {\thl{\Gamma}, \phl x: \thl {A \parr B}}
      }
      {\judge{}{\phl P}{\thl{\Gamma}, \phl y:\thl A, \phl x:\thl B}}
	\\\\
      \infer[\one]
      {\judge{}{\phl{\close x}}{\phl x: \thl\one}}
      { }
      \qquad
      \infer[\bot]
      {\judge{}{\phl {\wait xP}}{\thl \Gamma, \phl x:\thl\bot}}
      {\judge{}{\phl P}{\thl\Gamma}}
      \qquad
            \infer[\rname{Cut}]
      {\judge{}{\phl{\cpres{x}A{y}{A^\perp}(P \pp Q)}}{\thl\Gamma,\thl\Delta}}
      {\judge{}{\phl P}{\thl\Gamma, \phl x:\thl A} & \judge{}{\phl Q}{\thl\Delta,\phl
y:\thl{\dual{A}}}}
      \end{array}
  \end{displaymath}
  \caption{Multiplicative \CP, Typing Rules.}
  \label{fig:mult_cp_rules}
\end{figure}
The rules are exactly those from multiplicative classical linear logic, and enforce linear usage of
names. Rule $\tensor$ types an output $\send xyPQ$ as $A\tensor B$ by checking that the
continuation $P$
will implement behaviour $A$ on $y$ and behaviour $B$ on $x$. Dually, rule $\parr$ types an input
$\recv
xyP$ by checking that the continuation $P$ implements behaviour $A$ on $y$ and behaviour $B$ on
$x$.
Rules $\one$ and $\bot$ are straightforward. Rule $\rname{Cut}$ types a parallel composition
$\cpres xAy{}
(P\pp Q)$ by using duality of types. Specifically, we check that the processes $P$ and $Q$
implement dual behaviour on the respective channels $x$ and $y$.

\paragraph{Semantics and Properties}
The culmination of \CP is the semantics that we obtain from its definition. Specifically, we
can reuse the metatheory of cut elimination for linear logic in order to get a semantics for processes
with nontrivial properties. Consider the following example of an application of $\rname{Cut}$.
\[
\infer[\rname{Cut}]
{
\judge{}{
\cpres{x}{A \tensor B}{y}{}
\left( \send x{x'}PQ
\pp \recv y{y'}R \right)
}{\Gamma,\Gamma',\Delta}
}
{
\infer[\tensor]
      {
        \judge{}{\phl{\send x{x'}PQ}}
        {\thl{\Gamma}, \thl{\Gamma'}, \phl x:\thl{A \tensor B}}
      }
      {
        \judge{}{\phl P}
        {\thl{\Gamma}, \phl {x'}:\thl A}
        &
        \judge{}{\phl Q}
        {\thl \Gamma', \phl x: \thl B}
      }
&
      \infer[\parr]
      {\judge{}{\phl{\recv y{y'}R}}
        {\thl{\Delta}, \phl {y}: \thl {\dual A \parr \dual B}}
      }
      {\judge{}{\phl R}{\thl{\Delta}, \phl {y'}:\dual A, \phl y:\dual B}}
}
\]
The definitions of $P$, $Q$, $R$, $\Gamma$, $\Delta$, $A$, and $B$ are irrelevant for this example: the following reasoning holds in general, assuming that the premises $\judge{}{\phl P}
        {\thl{\Gamma}, \phl {x'}:\thl A}$, $\judge{}{\phl Q}
        {\thl \Gamma', \phl x: \thl B}$, and $\judge{}{\phl R}{\thl{\Delta}, \phl {y'}:\dual A, \phl
y:\dual B}$ are derivable.
The reader familiar with linear logic may have recognised that this example is a case of the proof
of cut
elimination. In general, cuts can always be eliminated from linear logic proofs, by following
an inductive proof-rewriting procedure. \CP inherits this property.
Here, we can rewrite the proof above to use cuts on smaller types, eliminating the $\tensor$
connective.
\begin{displaymath}
\infer[\rname{Cut}]
{
\judge{}{\cpres{x'}{A}{y'}{}(P \pp\ \cpres{x}{B}{y}{} (Q \pp R)\ )}
{\Gamma,\Gamma',\Delta}
}{
\judge{}{\phl P}
        {\thl{\Gamma}, \phl {x'}:\thl A}
&
\infer[\rname{Cut}]
{\judge{}{\cpres{x}{B}{y}{} (Q \pp R)}{\Gamma',\Delta, y':\dual A}}
{
	\judge{}{\phl Q}
        {\thl \Gamma', \phl x: \thl B}
	&
		\judge{}{\phl R}{\thl{\Delta}, \phl {y'}:\dual A, \phl y:\dual B}
}
}
\end{displaymath}
This is called a cut reduction.
Observe what happened to the initial process. We have rewritten it as follows.
\begin{equations}
	\phl{\cpres{x}{A\tensor B}{y}{\dual A\parr \dual B}
        (\send x{x'}PQ  \pp \recv y{y'}R)}
        &\reducesto&
        \phl{\cpres{x'}{A}{y'}{\dual A}
        (P  \pp \cpres xBy{\dual B} (Q\pp R))}
\end{equations}

This reduction rule executes a communication between two compatible processes composed in
parallel. The remarkable aspect of \CP is that we obtained the reduction from something that we 
already know from linear logic, which is the proof transformation that we discussed: we just need 
to observe what happens
to the processes. Furthermore, it is evident from the construction of the transformation that it
preserves types. All rules for the operational semantics of \CP are built by following this method.
We report here the principal reductions for multiplicative \CP, and the rules for allowing
reductions
inside of cuts.
\begin{equationsl}
	\phl{\cpres{x}{A\tensor B}{y}{\dual A\parr \dual B}
        (\send x{x'}PQ  \pp \recv y{y'}R)}
        &\reducesto &
        \phl{\cpres{x'}{A}{y'}{\dual A}
        (P  \pp \cpres xBy{\dual B} (Q\pp R))}\\
        \phl{\cpres x{\one}y\perp(\close x\pp \wait yP)}
        &\reducesto &
        \phl{P}\\
    \cpres{x}{A}{y}{}(P_1 \pp Q)
    & \reducesto & \cpres{x}{A}{y}{}(P_2 \pp Q)
    & \text{if } P_1 \reducesto P_2
    \\
    \cpres{x}{A}{y}{}(P \pp Q_1)
    & \reducesto & \cpres{x}{A}{y}{}(P \pp Q_2)
    & \text{if } Q_1 \reducesto Q_2
\end{equationsl}

The way in which the semantics of \CP is defined gives us two results on a silver platter.
First, all reductions in \CP preserve typing (type preservation), because cut elimination preserves
types. Second, all sessions are eventually completely executed (progress), because all cuts can be
eliminated from proofs.

\subsection{From \CP to \CHOP}
\label{sec:from_cp_to_chop}

We attempt at applying the same method used for the construction of \CP to extend it to process
mobility with higher-order primitives.
More specifically, we need to find out how we can extend the proof theory of \CP with rules that
(i) type new process terms that capture the feature of process mobility, and (ii) enable proof
transformations that yield their expected semantics.

\paragraph{Processes}
What kind of process terms enable process mobility? For this, we can get direct inspiration from
\HOpi. We need three basic terms.
\begin{itemize}
\item A term for sending process code over a channel.
\item A term for receiving process code over a channel and storing it in a process variable.
\item A term for running the process code stored in a variable.
\end{itemize}
We thus extend the syntax of multiplicative \CP as follows, where $p,q,r,\ldots$ range over process
variables.
\begin{displaymath}
\begin{array}{rll}
P,Q,R ::= & \ldots & \\
& \sendho x{\abstr \rho P} & \text{higher-order output} \\
& \recv xpQ & \text{higher-order input} \\
& \invoke p\rho & \text{run}
\end{array}
\end{displaymath}

The higher-order output term $\sendho x{\abstr \rho P}$ sends the \emph{abstraction} $\abstr \rho
P$ over channel $x$. An abstraction is a parameterised process term, enabling the reuse of process
code in different contexts. The concept of abstraction is standard in the literature of \HOpi
\citep{SW01}; here, we apply a slight twist and use \emph{named} formal parameters $\rho$ rather
than just a parameter list, to make the invocation of abstractions not dependent on the order in
which actual parameters are passed. This plays well with the typing contexts of linear logic, used
for typing process terms in \CP, since they are order-independent as well (exchange is allowed).
We use $l$ to range over parameter names. These are constants, and are thus not affected by
$\alpha$-renaming. Formally, $\rho$ is a record that maps parameter names to channels
used inside $P$, i.e., $\rho = \{ l_i = x_i \}_i$. We omit curly brackets for records in the 
remainder when they are clear from the context, e.g., as in $\invoke p {l_1 = x, l_2 = y}$.
Given a record $\rho = \{ l_i = x_i \}_i$, we call the set $\{l_i\}_i$ the preimage of $\rho$ and 
the set $\{x_i\}$ the image of
$\rho$. An abstraction $\abstr \rho P$ binds all names in the image of $\rho$ in $P$. We 
require abstractions to be closed with respect to channels, in the sense that the image of $\rho$ 
needs to be exactly the set of all free channel names in $P$.

Dual to higher-order output, the higher-order input term $\recv xpQ$ receives an abstraction over
channel $x$ and stores it in the process variable $p$, which can be used in the continuation $Q$.

Abstractions stored in a process variable $p$ can be invoked (we also say run) using term
$\invoke p\rho$, where $\rho$ are the actual named parameters to be used by the process.

We purposefully chose a minimalistic syntax for \CHOP, for economy of the calculus. For example, 
our process output term has no continuation, but we will see in \cref{sec:proc_output_wc} that an 
equivalent construct with continuation can be easily given as
syntax sugar. The continuation in the process input term $\recv xpQ$ is necessary for scoping,
because $p$ is bound in $Q$. The syntax of the run term is clearly minimal.

\begin{remark}
The choice of having (named) parameters and abstractions is justified by the desire to use received
processes to implement behaviours. For example, we may want to receive a channel and then use a
process to implement the necessary behaviour on that channel: $\recvs xy.\recvs xp.\invoke p{l=y}$.

A different way of achieving the same result would be to support the communication of
processes with free names (not abstractions), and then allow run terms to dynamically rename free
names of received processes. For example, if we knew from typing that any process received for $p$
above has $z$ as free names, we could write $\recvs xy.\recvs xp.\invoke p{x=y}$.

Such a \emph{dynamic binding} mechanism would make our syntax simpler, since invocations would
simply be $\invoke p\sigma$ ($\sigma$ is a name substitution) and the higher-order output term
would be $\sendho xP$. However, dynamic binding is undesirable in a programming model, since the
scope of free names can change due to higher-order communications. The reader unfamiliar
with dynamic binding may consult \citep[p.~376]{SW01} for a discussion of this issue in \HOpi.
\end{remark}

\paragraph{Typing}
We now move to typing. The most influential change is actually caused by the run term, so we start
from there.
How should we type a term $\invoke p\rho$? The rule should look like the following informal
template (where
$??$ are unknown).
\[
\infer[\rname{}]
{\judge{}{\invoke p\rho}{\ ??}}
{??}
\]
Since $p$ stores a process (abstraction), and processes communicate over channels, we should be
able to type it with some environment $\Gamma$.
But what $\Gamma$ should we use? This depends on what
abstraction is stored in variable $p$. Such way of reasoning immediately calls for the introduction
of a new \emph{higher-order} context, which we denote with $\Theta$, where we store typing 
information on process variables. This
is the usual typing of variables seen in many calculi, but instead of value types, we store process
types.
The type of a process in \CP is evident: it is a context $\Gamma$.
So a higher-order context $\Theta$ associates process variables to process types $\Gamma$, i.e.,
$\Theta = p_1:\Gamma_1,\ldots,p_n:\Gamma_n$. Notice, however, that since each $p_i$ is going to be 
instantiated by an abstraction, we are not typing channels in each $\Gamma_i$, but the names of 
formal parameters ($l$).
Thus we can write the following rule.
\[
	\infer[\rname{Id}] {
		\judge{p:\Gamma}{\invoke p\rho}{\Gamma\rho}
	}{}
\]
Rule \rname{Id} says that if we run $p$ by providing all the necessary
channel parameters (collected in $\rho$) for running the abstraction that it
contains, then we implement exactly the session types that type the code that may be contained
in $p$ (modulo the renaming of the formal parameters performed with $\rho$, hence the 
$\Gamma\rho$). Notice that we use \emph{all} our available ``resources''---the channel names in the 
image of $\rho$---to run $p$, passing them as parameters. This ensures linearity.

Our rule means, however, that we have changed the shape of typing judgements from that in \CP,
since we have added $\Theta$. Luckily, adapting all rules in \CP to the new format is
straightforward: we
just need to distribute $\Theta$ over premises, respecting linearity of process variables. For
example, rule $\rname{Cut}$ is adapted as follows.
\[
\infer[\rname{Cut}]
	{\judge{\Theta,\Theta'}{\phl{\cpres{x}A{y}{A^\perp}(P \pp Q)}}{\thl\Gamma,\thl\Delta}}
	{\judge{\Theta}{\phl P}{\thl\Gamma, \phl x:\thl A} & \judge{\Theta'}{\phl Q}{\thl\Delta,\phl
	y:\thl{\dual{A}}}}
\]
The full set of adapted rules is given in the next section.
We read the judgement $\judge{\Theta}{P}{\Gamma}$ as ``$P$ uses process variables
according to $\Theta$ and channels according to $\Gamma$''. Observe that some channel behaviour may
not be implemented by $P$ itself, but it could instead be delegated to some invocations of process
variables as we exemplified before.

Now that we know how to type process variables, we can move on to typing process inputs and
outputs. Before we do it, though, we need to extend the syntax of types to capture the sending and
receiving of processes, as follows.
\begin{displaymath}
\begin{array}{rll}
A,B,C ::= & \ldots & \\
& \provide{\Gamma} & \text{send process of type $\Gamma$} \\
& \assume{\Gamma} & \text{receive process of type $\Gamma$}
\end{array}
\end{displaymath}

We can now type a process input.
\[
\infer[\rname{\assume{}}]
      {\judge{\Theta}{\recv xpP}{\Gamma, x:\assume{\Delta}}}
      {\judge{\Theta,p:\Delta}{P}{\Gamma}}
\]
Rule $\assume{}$ says that if we receive a process of type $\Delta$ over channel $x$ and
store it in variable $p$, we can use variable $p$ later in the continuation $Q$ assuming that it
implements $\Delta$. The type for the channel $x$ is $\assume{\Delta}$. If we think
in terms of channels, this type means that we receive a process of type $\Delta$. If we
think in terms of proof theory, it means that we are assuming that there exists of a proof 
of $\Delta$.

The rule for typing a process output follows.
\[
\infer[\rname{\provide{}}]
      {\judge{\Theta}{\sendho x{\abstr \rho P}}{x: \provide{\Gamma}}}
      {\judge{\Theta}{P}{\Gamma\rho}}
\]
Rule $\provide{}$ should be self-explanatory by now.
It says that if we send a process of type $P$ over channel $x$ and $P$ implements $\Gamma$,
then $x$ has type $\provide{\Gamma}$. Note that we allow (actually, we ensure, because of
linearity) $P$ to use the process variables available in the context ($\Theta$). This is not
an arbitrary choice made just to increase expressivity of the calculus: we will see in
\cref{sec:chop} that it follows naturally from the notion of $\eta$-expansion of \CLL.

\paragraph{Usage and Semantics}
Intuitively, the type constructors $\provide{}$ and $\assume{}$ are dual of each other: for any
${\Gamma}$, $\assume{\Gamma}$ means that we need a proof of $\Gamma$,
while $\provide{\Gamma}$ means
that we provide such a proof. So we extend duality as follows.
\begin{displaymath}
\begin{array}{rcl}
\dual{\provide{\Gamma}} & = & \assume{\Gamma} \\
\dual{\assume{\Gamma}} & = & \provide{\Gamma}
\end{array}
\end{displaymath}

This allows us to compose processes that make use of the higher-order input and output primitives.
For example, consider the following cut.
\begin{displaymath}
\infer[\rname{Cut}]{
		\judge{\Theta,\Theta'}{\cpres x{\provide{\Delta}}y{} ( \sendho x{\abstr \rho P} \pp
\recv xpQ
)}{\Gamma}
	}{
		\infer[\rname{\provide{}}]{
			\judge{\Theta}{\sendho x{\abstr \rho P}}{x:\provide{\Delta}}
		}{\judge{\Theta}{P}{\Delta\rho}}
	&
	\infer[\rname{\assume{}}]
	{
		\judge{\Theta'}{\recv xpQ}{\Gamma, y:\assume{\Delta}}
	}{
		\judge{\Theta', p: \Delta}{Q}{\Gamma}
	}
	}
\end{displaymath}

Let us now consider what the operational semantics of term $\cpres x{\provide{\Delta}}y{} (
\sendho x{\abstr \rho P} \pp
\recv xpQ
)$ should be. Intuitively, we would expect
it to reduce to something
like $Q$
with each occurrence of $p$ replaced by $P$. Indeed, in \HOpi (adapted to our syntax), we would
have
$\cpres x{\provide{\Delta}}y{} ( \sendho x{\abstr \rho P} \pp
\recv xpQ ) \reducesto Q\{\abstr \rho P/p\}$. In our setting,
using this kind of ``immediate'' substitution can be undesirable for two reasons. First, it is
impractical: in practice, substitutions are typically implemented by replacing subterms as needed
(possibly in a delayed fashion),
rather than blind immediate rewriting \citep{W71}. Second, it is theoretically inconvenient: in our
example
of cut reduction in multiplicative \CP, we could prove that the reductum is well-typed just by
looking at
the premises of the cut (more precisely, using them to write a new smaller proof); but here it
is not possible, since our premises do not talk about $Q\{\abstr \rho P/p\}$.

We deal with both issues by introducing a notion of explicit substitution for process variables.
The syntax of an explicit substitution is $Q \letp p{\abstr \rho P}$, read ``let $p$ be $\abstr
\rho P$ in $Q$''. We type explicit substitutions with the rule below.
\begin{displaymath}
\infer[\rname{Chop}]{
		\judge{\Theta,\Theta'}{Q \letp p{\abstr \rho P}}{\Gamma}
	}{
		\judge{\Theta}{P}{\Delta\rho}
		&
		\judge{\Theta',p:\Delta}{Q}{\Gamma}
	}
\end{displaymath}
Rule \rname{Chop} allows us to replace $p$ with $\abstr \rho P$ in $Q$, provided that $P$ and
$q$ have compatible typing (up to the formal named parameters). If you think in terms of
processes, \CHOP stands for ``Cut for Higher-Order Processes''. If you
think in terms of logic, \CHOP stands for ``Cut for Higher-Order Proofs''. The idea is that a
variable $p$ stands for a ``hole'' in a proof, which has to be filled as expected by the type for
$p$.

Observe that our previous proof for typing $\cpres x{\provide{\Delta}}y{} ( \sendho x{\abstr
\rho P} \pp
\recv xpQ )$ has exactly the same premises needed for rule $\rname{Chop}$. This means that we can
derive
the following reduction rule.
\[
\cpres x{\provide{\Delta}}y{} ( \sendho x{\abstr \rho P} \pp
\recv xpQ )
\reducesto
Q \letp p{\abstr \rho P}
\]

The reduction above models exactly what we would expect in a higher-order process calculus, but
using an explicit substitution. In \cref{sec:chop}, we will see that rule \rname{Chop} allows for
sound proof rewrites that yield the expected semantics of explicit substitutions for higher-order
processes. In particular, explicit substitutions can be propagated inside of terms, as illustrated
by the following example reduction.
\begin{equations}
(\recv xyQ) \letp p{\abstr \rho P} & \reducesto & \recv xy{(Q \letp p{\abstr \rho P})}
\end{equations}
Also, when a substitution reaches its target process variable, it replaces it.
\begin{equations}
\invoke p\rho \letp p{\abstr {\rho'} P} & \reducesto & P\{\rho \circ {\rho'}^{-1} \}
\end{equations}
Above, we abuse notation and interpret $\rho$ as functions that map labels to names, and
then $\circ$ is standard function composition.

This ends our overview of the intuition behind the journey from \CP to \CHOP. In the next
sections, we move to the formal presentation of \CHOP, its metatheory, extensions, and further
discussions.

\section{Classical Higher-Order Processes (\CHOP)}
\label{sec:chop}

We formally introduce the calculus of Classical Higher-Order Processes (\CHOP).
Technically, \CHOP is a strict generalisation of the calculus of Classical Processes (\CP) by
\citet{W14}. We refer to the latest version of \CP, given in \citep{CLMSW16}.

\paragraph{Processes}
In \CHOP, programs are processes ($P$,$Q$,$R$,$\ldots$) that communicate over channels
($x$,$y$,$\ldots$).
The syntax of processes is the following. Some terms include types ($A,B,C,\ldots$), which we
present
afterwards.
\begin{displaymath}
\begin{array}{rll}
P,Q,R ::=
& \send xyPQ & \text{output a channel} \\
& \recv xyP & \text{input a channel} \\
& \inl xP & \text{select left} \\
& \inr xP & \text{select right} \\
& \piCase{x}{P}{Q} & \text{offer a choice} \\
& \client{x}{y}{P} & \text{client request} \\
& \srv{x}{y}{P} & \text{server accept} \\
& \sendtype{x}{A}{P} & \text{output a type} \\
& \recvtype{x}{X}{P} & \text{input a type} \\
& \sendho{x}{\abstr \rho P} & \text{higher-order output} \\
& \recv xpP & \text{higher-order input} \\
& \invoke p\rho & \text{run} \\
& \close x & \text{close} \\
& \wait xP & \text{wait} \\
& \emptyCase x & \text{empty offer} \\
& \cpaxiom{x}{}{y}{A} & \text{link} \\
& \cpres{x}{A}{y}{}(P \pp Q) & \text{parallel composition} \\
& P \letp q{\abstr \rho Q} & \text{explicit substitution}
\end{array}
\end{displaymath}
We are already familiar with some terms, from the overview in \cref{sec:preview}. Here, we
briefly describe the meaning of each.

Term $\send xyPQ$ sends $y$ over $x$, and then continues as $P$ and $Q$ in parallel. Dually, term
$\recv xyP$ receives $y$ over $x$. In both these terms, $y$ is bound in the continuation $P$.
Term $\inl xP$ selects the left branch of an offer available over $x$. Likewise, $\inr xP$ selects
the
right branch. Dually, term $\piCase xPQ$ offers a choice between $P$ (left branch) and $Q$ (right
branch)
over $x$. Term $\client xyP$ sends $y$ to a server (an always-available replicated process)
available
over $x$, and then proceeds as $P$. Dually, term $\srv xyP$ is a replicated server waiting for
requests to
create new instances on $x$ (the instance will then communicate with the client over $y$). Name $y$
is
bound in both the client and server terms.
Terms $\sendtype{x}{A}{P}$ and $\recvtype xXP$ enable polymorphism.
Term $\sendtype{x}{A}{P}$ sends type $A$ over $x$ and proceeds as $P$; while term $\recvtype xXP$
receives a type over $x$ and then proceeds as $P$, where $X$ is bound in $P$.
Term $\sendho x{\abstr \rho P}$ sends the (process) abstraction $\abstr \rho P$ over $x$, where
$\rho$ is a record that maps the parameter names (ranged over by $l$) of the abstraction to free names in $P$. Dually,
term $\recv xpP$ receives an abstraction over $x$ and stores it in variable $p$, to be used in the
continuation $P$ ($p$ is bound in $P$). A term $\invoke p\rho$ runs the process stored inside of
variable $p$, instantiating its named parameters according to $\rho$.
Term $\close x$, or empty output, sends over $x$ the message that the channel is being closed and
closes it. Term $\wait xP$ waits for one of such close messages on $x$ and then continues as $P$.
A link term $\cpaxiom x{}yA$ acts as a forwarding proxy: inputs along $x$ are forwarded as outputs
to $y$ and vice versa, any number of times (as dictated by the protocol $A$).
The parallel composition term $\cpres{x}{A}{y}{}(P \pp Q)$ connects the channels $x$ at
process $P$ and $y$ at process $Q$ to form a session, enabling the two processes to communicate.
The names $x$ and $y$ are bound, respectively, in $P$ and $Q$. The type $A$ describes the behaviour
that the processes will follow when communicating.
Finally, the explicit substitution term $P \letp q{\abstr \rho Q}$ stores the abstraction $\abstr
\rho Q$ in variable $q$, and binds the latter in $P$.

$\alpha$-conversion works as usual. Following the intuition that invocation terms use named
parameters, bound names inside of $\rho$ in terms $\invoke p\rho$ are also affected by $\alpha$-conversion.
For example, these two processes are $\alpha$-equivalent: $\recvs xy.\invoke p{l = y}$ and
$\recvs xz.\invoke p{l = z}$.
Parameters are technically important: in the second process, if we did not rename $y$ to $z$
in the substitution, $y$ would become a free name but it was a bound name in the first process!
This is the same mechanism behind $\alpha$-equivalence for abstraction invocation in \HOpi, just extended to our named parameters.
We consider processes up to $\alpha$-equivalence in the remainder.

\paragraph{Types and environments}
There are two kinds of types in \CHOP. Session types, ranged over by $A, B, C, \ldots$, are used to 
type channels, and process types, ranged over by $\Gamma,\Delta,\ldots$, are used to type processes.

We inherit all the session types of \CP, which correspond to propositions in \CLL, and have also our
new ones for typing communication of processes. We range over
atomic propositions in session types with $X, Y$.
The syntax of types is given in the following, along with a brief description of
each case. We use $s$ to range over parameter names ($l$) and channels ($x$).
\begin{displaymath}
\begin{array}{rll}
A,B,C ::= & \gentensor A{}B & \text{send $A$, proceed as $B$} \\
& \genparr A{}B & \text{receive $A$, proceed as $B$}\\
& \genoplus A{}B & \text{select $A$ or $B$}\\
& \genwith A{}B & \text{offer choice between $A$ and $B$}\\
& \zero & \text{unit for $\oplus$}\\
& \top & \text{unit for $\with$}\\
& \one & \text{unit for $\tensor$}\\
& \bot & \text{unit for $\parr$}\\
& ?A & \text{client request}\\
& !A & \text{server accept}\\
& \exists X.A & \text{existential}\\
& \forall X.A & \text{universal}\\
& X & \text{atomic proposition}\\
& \dual X & \text{dual of atomic proposition}\\
& \provide{\Gamma} & \text{send process of type $\Gamma$} \\
& \assume{\Gamma} & \text{receive process of type $\Gamma$}
\\\\
\Gamma,\Delta ::= & s_1:A_1, \ldots, s_n:A_n & \text{process type}
\end{array}
\end{displaymath}

In $\exists X.A$ and $\forall X.A$, type variable $X$ is bound in $A$. We write $\m{ftv}(A)$ for
the set of free type variables in $A$. We ignore
the order of associations in process types (exchange is allowed). Whenever we write $\Gamma,\Delta$,
we assume that $\Gamma$ and $\Delta$ are disjoint (share no channel names). We write $(s_i:A_i)_i$ 
as an abbreviation of $s_1:A_1,\ldots,s_n:A_n$. Typing will ensure that all keys
in an environment are of the same kind (either parameter names or channels).

Just like \CP, \CHOP uses the standard notion of duality from linear logic to check that session
types are compatible. As described in \cref{sec:from_cp_to_chop}, we extend duality to deal with our
new session types, $\provide{\Gamma}$ and $\assume{\Gamma}$. Formally, we write
$\dual A$ for the type dual to $A$, defined inductively as follows.
\begin{displaymath}
\begin{array}{rcl@{\qquad}rcl}
\dual{(X)} &=& \dual X & \dual{(\dual X)} &=& X
\\
\dual{(\gentensor A{}B)} &=& \genparr{\dual A}{}{\dual B}
&
\dual{(\genparr A{}B)} &=& \gentensor{\dual A}{}{\dual B}
\\
\dual{(\genoplus A{}B)} &=& \genwith{\dual A}{}{\dual B}
&
\dual{(\genwith A{}B)} &=& \genoplus{\dual A}{}{\dual B}
\\
\dual\zero &=& \top & \dual\top &=& \zero
\\
\dual\one &=& \bot & \dual\bot &=& \one
\\
\dual{(?A)} &=& !A & \dual{(!A)} &=& ?A
\\
\dual{(\exists X.A)} &=& \forall X.\dual A
&
\dual{(\forall X.A)} &=& \exists X.\dual A
\\
\dual{\provide{\Theta;\Gamma}} &=& \assume{\Theta;\Gamma}
&
\dual{\assume{\Theta;\Gamma}} &=& \provide{\Theta;\Gamma}
\end{array}
\end{displaymath}
Duality is an involution, i.e., $\dual{(\dual A)} = A$.

A process environment $\Theta$ associates process variables to process types.
\begin{displaymath}
\begin{array}{rll}
\Theta ::= & p_1:\Gamma_1, \ldots, p_n:\Gamma_n & \text{process environment}
\end{array}
\end{displaymath}
We write $\cdot$ for the empty process environment (no associations). Similarly to process types, 
we allow for exchange in process environments and require all process variables in a process 
environment to be distinct.

\paragraph{Typing}
Typing judgements in \CHOP have the form $\judge{\Theta}{P}{\Gamma}$, which reads ``process $P$
uses process variables according to $\Theta$ and channels according to $\Gamma$''. We
omit $\Theta$ in judgements when it is empty ($\cdot$).
The rules for deriving judgements are displayed in \cref{fig:typing_rules}.
\begin{figure}[t]
  \begin{displaymath}
    \begin{array}{c}
      \infer[\rname{Axiom}]
      {\judge{}{\phl{\cpaxiom{x}{A^\perp}{y}{A}}}{\phl x:\thl {A^\bot}, \phl y:\thl{A}}}
      {}
      \qquad
      \infer[\rname{Cut}]
      {\judge{\Theta,\Theta'}{\phl{\cpres{x}A{y}{A^\perp}(P \pp Q)}}{\thl\Gamma,\thl\Delta}}
      {\judge{\Theta}{\phl P}{\thl\Gamma, \phl x:\thl A} & \judge{\Theta'}{\phl Q}{\thl\Delta,\phl
y:\thl{\dual{A}}}}
      \\\\
      \infer[\tensor]
      {
        \judge{\Theta,\Theta'}{\phl{\send xyPQ}}
        {\thl{\Gamma}, \thl{\Delta}, \phl x:\thl{A \tensor B}}
      }
      {
        \judge{\Theta}{\phl P}
        {\thl{\Gamma}, \phl y:\thl A}
        &
        \judge{\Theta'}{\phl Q}
        {\thl \Delta, \phl x: \thl B}
      }
      \qquad
      \infer[\parr]
      {\judge{\Theta}{\phl{\recv xyP}}
        {\thl{\Gamma}, \phl x: \thl {A \parr B}}
      }
      {\judge{\Theta}{\phl P}{\thl{\Gamma}, \phl y:\thl A, \phl x:\thl B}}
      \\\\
      \infer[\oplus_1]
      {\judge{\Theta}{\phl{\inl xP}}{\thl\Gamma, \phl x:\thl{A \oplus B}}}
      {\judge{\Theta}{\phl P}{\thl\Gamma, \phl x:\thl A}}
      \qquad
      \infer[\oplus_2]
      {\judge{\Theta}{\phl{\inr xP}}{\thl\Gamma, \phl x:\thl{A \oplus B}}}
      {\judge{\Theta}{\phl P}{\thl\Gamma, \phl x:\thl B}}
      \qquad
      \infer[\with]
      {\judge{\Theta}{\phl{\piCase xPQ}}{\thl \Gamma, \phl x:\thl{A \with B}}}
      {\judge{\Theta}{\phl P}{\thl\Gamma, \phl x:\thl A} &
      \judge{\Theta}{\phl Q}
        {\thl\Gamma, \phl
        x:\thl B}}
       \\\\
       \infer[?]
       {\judge{\Theta}{\phl{\client xyP}}{\thl \Gamma, \phl x:\thl {\query A}}}
       {\judge{\Theta}{\phl P}{\thl \Gamma, \phl y:\thl A}}
       \qquad
       \infer[!]
       {\judge{}{\phl{\srv xyP}}{\thl{\query \Gamma}, \phl x:\thl{\bang A}}}
       {\judge{}{\phl P}{\thl{\query \Gamma}, \phl y:\thl A}}
       \\\\
       \infer[\exists]
       {\judge{\Theta}{\phl{\sendtype xAP}}{\thl \Gamma, \phl x:\thl{\exists X.B}}}
       {\judge{\Theta}{\phl P}{\thl \Gamma, \phl x:\thl{B\{A/X\}}}}
       \qquad
       \infer[\forall]
       {\judge{\Theta}{\phl{\recvtype xXP}}{\thl\Gamma, \phl x:\thl{\forall X.B}}}
       {\judge{\Theta}{\phl P}{\thl\Gamma, \phl x:\thl B & \thl X \notin \m{ftv}(\thl \Theta) \cup
\m{ftv}(\Gamma)}}
       \\\\
       \infer[\rname{Weaken}]
       {\judge{\Theta}{\phl P}{\thl \Gamma, \phl x:\thl{\query A}}}
       {\judge{\Theta}{\phl P}{\thl \Gamma}}
       \qquad \infer[\rname{Contract}]
       {\judge{\Theta}{\phl{P\{x/y, x/z\}}}{\thl\Gamma, \phl x:\thl{\query A}}}
       {\judge{\Theta}{\phl P}{\thl\Gamma, \phl y:\thl{\query A}, \phl
         z:\thl{\query A}}}
      \\\\
      \infer[\one]
      {\judge{}{\phl{\close x}}{\phl x: \thl\one}}
      { }
      \qquad
      \infer[\bot]
      {\judge{\Theta}{\phl {\wait xP}}{\thl \Gamma, \phl x:\thl\bot}}
      {\judge{\Theta}{\phl P}{\thl\Gamma}}
       \quad\qquad
       \text{(no rule for $\zero$)}
      \qquad
      \infer[\top]
      {\judge{\Theta}{\phl{\emptyCase x}}{\thl{\Gamma}, \phl x:\thl\top}}
      {}
      \\\\
      \infer[\rname{Id}]
      {
      	\judge{p:\Gamma}{\invoke p\rho}{\Gamma\rho}
      }
      {}
      \qquad
      \infer[\rname{Chop}]
      {
		\judge{\Theta,\Theta'}{Q \letp p{\abstr \rho P}}{\Gamma}
      }
      {
        \judge{\Theta}{P}{\Delta\rho}
        &
        \judge{\Theta', p:\Delta}{Q}{\Gamma}
      }
      \\\\
      \infer[\rname{\provide{}}]
      {\judge{\Theta}{\sendho x{\abstr \rho P}}{x: \provide{\Gamma}}}
      {\judge{\Theta}{P}{\Gamma\rho}}
      \qquad
      \infer[\rname{\assume{}}]
      {\judge{\Theta}{\recv xpP}{\Gamma, x:\assume{\Delta}}}
      {\judge{\Theta,p:\Delta}{P}{\Gamma}}
      \end{array}
    \end{displaymath}
  \caption{\CHOP, Typing Rules.}
  \label{fig:typing_rules}
\end{figure}

Typing rules associate types to channels by looking at how the
channel is used in process terms, as expected. They should be unsurprising by now, after our
discussion in \cref{sec:from_cp_to_chop}. We just make a few observations for the reader unfamiliar
with \CP and/or linear logic. Rule \rname{Axiom} types a link between $x$ and $y$ by requiring that
the types of $x$ and $y$ are the dual of each other. This ensures that any message on $x$ can be
safely forwarded to $y$, and vice versa. All rules for typing channels enforce linear usage aside
for client requests (typed with the exponential connective $?$), for which contraction and
weakening are allowed. Contraction (rule \rname{Contract}) allows for having multiple client requests for the same
server channel, and weakening (rule \rname{Weaken}) allows for having a client that does not use a server.
Rule $!$ types a server, where $\query\Gamma$ denotes that all session types in $\Gamma$ must be client
requests, i.e., $\query \Gamma ::= x_1:\query A_1, \ldots, x_n: \query A_n$.
A server must be executable any number of times, since it does not know how many client requests it will
have to support. This is guaranteed by requiring that all resources used by the server are acquired
by communicating with the client (according to protocol $A$) and with other servers ($?\Gamma$).

The rules for typing higher-order terms have already been discussed in \cref{sec:from_cp_to_chop}.
We can now substantiate our remark on the design of rule $\provide{}$ made in the same section, 
which allows the sent abstraction to use process variables available in the context of the sender 
($\Theta$). It is well-known that restricting rule \rname{Axiom} to atomic propositions in 
\CLL does not reduce expressivity, because the general \rname{Axiom} can be reconstructed as an 
admissible rule. A cut-free derivation of $\judge{\ }{\cpaxiom{x}{\dual A}{y}{A}}{x:\dual A,y:A}$ 
corresponds to the $\eta$-expansion of the link term $\cpaxiom{x}{\dual A}{y}{A}$, and 
$\eta$-expansion is also a convenient tool for the elegant formulation of multiparty sessions in 
\CP, as shown by \citep{CLMSW16}.
That the general \rname{Axiom} is admissible is thus an important property of \CP, and preserving 
this property is desirable in \CHOP for the same reasons.
The following proof shows the new case of $\eta$-expansion for higher-order communications 
introduced by our new types, which yields the expansion
$\cpaxiom{x}{\assume{\Delta}}{y}{\provide{\Delta}} \reducesto
x(p).x[\abstr \rho {\invoke p \rho}]$ (where $\rho$ captures all the names in $\Delta$)
.
\[
\infer[\assume{}]{
	\judge{}{x(p).x[\abstr \rho {\invoke p \rho}]}{x:\assume{\Delta}, 
x:\provide{\Delta}}
}{
	\infer[\provide{}]{
		\judge{p:\Delta}{x[\abstr \rho {\invoke p \rho}]}{x:\provide{\Delta}}
	}{
		\infer[\rname{Id}]{
			\judge{p:\Delta}{\invoke p \rho}{\Delta\rho}
		}{}
	}
}
\]
Notice that without allowing the sent abstraction to use the linear process variables available to 
the sender, the derivation would be invalid.
This also shows that rule \rname{Id} is in harmony with our rules for higher-order input and 
output. We still have to prove that our rules are in harmony also with rule \rname{Cut}, as 
usual, in order to obtain cut elimination. We prove this in \cref{sec:semantics}, and establish a 
metatheory for \CHOP in \cref{sec:meta}.

\subsection{Example}
\label{sec:cloud_server}
We illustrate the expressivity of \CHOP by implementing a cloud server for running applications
that require a database.
The idea is that clients are able to choose between two options: run both the application and the
database it needs in the server, or run just the application in the server and connect it to an
externally-provided database (which we could imagine is run somewhere else in the cloud).

\paragraph{A cloud server}
The process for the cloud server follows. We assume that $A$ is the protocol (left
unspecified) that applications have to use in order to communicate with databases, and $l$
is the named parameter used by both abstractions to represent this shared connection
(the parameter names may be different, we do this only for a simpler presentation).
We also let applications and databases access some external services, e.g., loggers, through the parameters
$\til l$ and $\til m$ respectively. We write $\til{l = x}$ for $l_1 = x_1,\ldots,l_n = x_n$.

\[
\begin{array}{l@{\quad}l}
\srv{cs}{x}{\ \piCase{x}{
&
	\recvs{x}{x'}. \recvs {x}{app}. \recvs {x'}{db}. \cpres{z}{A}{w}{}(
    \invoke{app}{l=z,\til{l = u}} \pp
    \invoke {db}{l=w,\til{m = v}} )
}{
\\ &
	\recvs{x}{extdb}. \recvs{x}{app}. \cpres{z}{A}{w}{}( \invoke{app}{l=z,\til{l = u}} \pp \cpaxiom{extdb}{}{w}{\dual A} )
\quad } }
\end{array}
\]

The cloud server waits for client requests on channel $cs$. Then, it communicates with the client
on the
established channel $x$. It offers the options that we mentioned through a choice, respectively
with the
left and right branches. In the left branch, we first receive an auxiliary channel $x'$. We then
receive
the application $app$ on $x$ and the database $db$ on $x'$. Then, we compose and run $app$ and $db$
in
parallel, connecting them through the endpoints $z$ and $w$.
In the right branch, we first receive a channel $extdb$ for communicating with the external
database, and
then receive the application $app$. Then, we compose $app$ with a link term
$\cpaxiom{extdb}{}{w}{\dual A}$, which connects the application to the external database.

\paragraph{Typing the cloud server}
\newcommand{\passume}[1]{(#1)}
We now illustrate how to use our typing rules to prove that the cloud server is well-typed.
Let $\query \Gamma = (l_i:\query A_i)_i$, $\query\Delta = (m_i:B_i)_i$, $\rho = \{l=z,\til{l = u}\}$, and $\rho' = \{l=w,\til{m = v}\}$.
For readability, we first type the left and right branches in the choice offered through $x$.
Here is the proof for the left branch.
\begin{displaymath}
\infer[\parr]
{
	\judge{}
	{\recvs{x}{x'}. \recvs {x}{app}. \recvs {x'}{db}. \cpres{z}{A}{w}{}(
  \invoke{app}{\rho} \pp \invoke {db}{\rho'} )
  }
	{\query\Gamma\rho,\query\Delta\rho',
			x: 
			\assume{\query\Delta, l:\dual A}
			\parr
			\assume{\query\Gamma, l:A}}
}
{
	\infer[\assume{}]
	{
		\judge{}
		{\recvs {x}{app}. \recvs {x'}{db}. \cpres{z}{A}{w}{}( \invoke{app}{\rho} \pp \invoke {db}{\rho'} )}
		{\query\Gamma\rho,\query\Delta\rho',
			x': \assume{\query\Delta, l:\dual A},
			x: \assume{\query\Gamma, l:A}}
	}
	{
		\infer[\assume{}]
		{
			\judge{app:\passume{\query\Gamma, l:A}}
			{\recvs {x'}{db}. \cpres{z}{A}{w}{}( \invoke{app}{\rho} \pp \invoke {db}{\rho'} )}
			{\query\Gamma\rho,\query\Delta\rho',
			x': \assume{\query\Delta, l:\dual A}}
		}
		{
			\infer[\rname{Cut}]
			{
				\judge{app:\passume{\query\Gamma, l:A}, db:\passume{\query\Delta, l:\dual A}}
				{\cpres{z}{A}{w}{}( \invoke{app}{\rho} \pp \invoke {db}{\rho'} )}
				{\query\Gamma\rho,\query\Delta\rho'}
			}
			{
				\infer[\rname{Id}]
				{
					\judge{app:\passume{\query\Gamma, l:A}}
					{\invoke{app}{\rho}}
					{\query\Gamma\rho, z:A}
				}{}
				&
				\infer[\rname{Id}]
				{
          \judge{db:\passume{\query\Delta, l:\dual A}}
          {\invoke{db}{\rho'}}
          {\query\Delta\rho', w:\dual A}
				}{}
			}
		}
	}
}
\end{displaymath}

And this is the proof for the right branch.
\begin{displaymath}
\infer[\parr]
{
	\judge{}
	{\recvs x{extdb}.\recvs {x}{app}. \cpres{z}{A}{w}{}(
  \invoke{app}{\rho} \pp \invoke {db}{\rho'} )
  }
	{\query\Gamma\rho,
			x: 
			A
			\parr
			\assume{\query\Gamma, l:A}}
}
{
	\infer[\assume{}]
	{
		\judge{}
		{\recvs {x}{app}. \cpres{z}{A}{w}{}( \invoke{app}{\rho} \pp \cpaxiom{extdb}{A}{w}{\dual A} )}
		{\query\Gamma\rho,
			extdb: A,
			x: \assume{\query\Gamma, l:A}}
	}
	{
		\infer[\rname{Cut}]
			{
				\judge{app:\passume{\query\Gamma, l:A}}
				{\cpres{z}{A}{w}{}( \invoke{app}{\rho} \pp \cpaxiom{extdb}{A}{w}{\dual A})}
				{\query\Gamma\rho,extdb:A}
			}
			{
				\infer[\rname{Id}]
				{
					\judge{app:\passume{\query\Gamma, l:A}}
					{\invoke{app}{\rho}}
					{\query\Gamma\rho, z:A}
				}{}
				&
				\infer[\rname{Axiom}]
				{
		          \judge{}
    		      {\cpaxiom{extdb}{A}{w}{\dual A}}
        		  {extdb:A, w:\dual A}	
        		}{}
			}
	}
}
\end{displaymath}

Now that we have proofs for the two branches, we can use them to type the entire cloud server.
Let $\prf L$ and $\prf R$ be the two proofs above, respectively, and $P_L$ and $P_R$ the processes
that
they type (the left and right branches in the cloud server).
Then, we type the cloud server as follows.
\begin{displaymath}
\infer[\bang]
{
	\judge{}{\srv{cs}{x}{\ \piCase{x}{P_L}{P_R}\ }}{
	\query\Gamma\rho, \query\Delta\rho', cs:\bang \left(
		\left( 
			\assume{\query\Delta, l:\dual A}
			\parr
			\assume{\query\Gamma, l:A} \right)
		\with
		\left( A \parr \assume{\query\Gamma, l:A} \right)\right)}
}
{
	\infer[\with]
	{
		\judge{}{\piCase{x}{P_L}{P_R}}
		{\query\Gamma\rho, \query\Delta\rho', x:
		\left( 
			\assume{\query\Delta, l:\dual A}
			\parr
			\assume{\query\Gamma, l:A} \right)
		\with
		\left( A \parr \assume{\query\Gamma, l:A} \right)}
	}
	{
		\prf L
		&
		\prf R
	}
}
\end{displaymath}

The types for our cloud server follows the intuition that we initially
discussed for this example. The typing derivation also illustrates the interplay between 
our new features (process mobility and the usage of process variables) with the other
features of the calculus---in this example: channel mobility, exponentials (replicated services), 
choices, and links.

\paragraph{A generic improvement}
The code of our cloud server implementation would not depend on how clients and database
communicate, were it not for the hardcoded protocol $A$.
We can get rid of the hardcoded $A$ and reach a generic implementation using polymorphism.
Here is the improved implementation, where we $\boxed{\text{box}}$ the improvements.
\[
\begin{array}{l@{\quad}l}
\srv{cs}{x}{\ \boxed{x(X).}\ \piCase{x}{
&
	\recvs{x}{x'}. \recvs {x}{app}. \recvs {x'}{db}. \cpres{z}{\boxed{X}}{w}{}(
    \invoke{app}{\rho} \pp
    \invoke {db}{\rho'} )
}{
\\ &
	\recvs{x}{extdb}. \recvs{x}{app}. \cpres{z}{\boxed{X}}{w}{}( \invoke{app}{\rho} \pp 
\cpaxiom{extdb}{}{w}{\boxed{\dual X}} )
\quad } }
\end{array}
\]
In the improved cloud server, the client must also send us the protocol $X$ that the application
will use to communicate with the database. The cloud server is thus now generic, and the type of 
channel $cs$ is the following.
\[
\bang \left(\ \boxed{\forall X.}\ \left(
		\left( 
			\assume{\query\Delta, l:\boxed{\dual X}}
			\parr
			\assume{\query\Gamma, l:\boxed{X}} \right)
		\with
		\left( \boxed{X} \parr \assume{\query\Gamma, l:\boxed{X}} \right)\right)\right)
\]

Here is the proof.
\begin{displaymath}
\infer[\bang]
{
	\judge{}{\srv{cs}{x}{\ x(X).\ \piCase{x}{P_L}{P_R}\ }}{
	\query\Gamma\rho, \query\Delta\rho', cs:\bang \left( \forall X.\ \left(
		\left( 
			\assume{\query\Delta, l:\dual X}
			\parr
			\assume{\query\Gamma, l:X} \right)
		\with
		\left( X \parr \assume{\query\Gamma, l:X} \right)\right)\right)}
}
{
	\infer[\forall]{
		\judge{}{x(X).\ \piCase{x}{P_L}{P_R}}{
	\query\Gamma\rho, \query\Delta\rho', x: \forall X.\ \left(
		\left( 
			\assume{\query\Delta, l:\dual X}
			\parr
			\assume{\query\Gamma, l:X} \right)
		\with
		\left( X \parr \assume{\query\Gamma, l:X} \right)\right)}
	}{
		\infer[\with]
		{
			\judge{}{\piCase{x}{P_L}{P_R}}
			{\query\Gamma\rho, \query\Delta\rho', x:
			\left( 
				\assume{\query\Delta, l:\dual X}
				\parr
				\assume{\query\Gamma, l:X} \right)
			\with
			\left( X \parr \assume{\query\Gamma, l:X} \right)}
		}
		{
			\prf L\{X/A\}
			&
			\prf R\{X/A\}
		}
	}
}
\end{displaymath}

The proofs $\prf L\{X/A\}$ and $\prf R\{X/A\}$ used above are as $\prf L$ and $\prf R$
respectively, but wherever we had $A$ we now have $X$.
The interplay between type variables and process variables merits illustration, so we show $\prf
R\{X/A\}$ in full.
\begin{displaymath}
\infer[\parr]
{
	\judge{}
	{\recvs x{extdb}.\recvs {x}{app}. \cpres{z}{X}{w}{}(
  \invoke{app}{\rho} \pp \invoke {db}{\rho'} )
  }
	{\query\Gamma\rho,
			x: 
			X
			\parr
			\assume{\query\Gamma, l:X}}
}
{
	\infer[\assume{}]
	{
		\judge{}
		{\recvs {x}{app}. \cpres{z}{X}{w}{}( \invoke{app}{\rho} \pp \cpaxiom{extdb}{X}{w}{\dual X} 
)}
		{\query\Gamma\rho,
			extdb: X,
			x: \assume{\query\Gamma, l:X}}
	}
	{
		\infer[\rname{Cut}]
			{
				\judge{app:\passume{\query\Gamma, l:X}}
				{\cpres{z}{X}{w}{}( \invoke{app}{\rho} \pp \cpaxiom{extdb}{X}{w}{\dual X})}
				{\query\Gamma\rho,extdb:X}
			}
			{
				\infer[\rname{Id}]
				{
					\judge{app:\passume{\query\Gamma, l:X}}
					{\invoke{app}{\rho}}
					{\query\Gamma\rho, z:X}
				}{}
				&
				\infer[\rname{Axiom}]
				{
		          \judge{}
    		      {\cpaxiom{extdb}{X}{w}{\dual X}}
        		  {extdb:X, w:\dual X}	
        		}{}
			}
	}
}
\end{displaymath}

Observe the application of rule $\assume{}$ in the proof. If we read it bottom-up, we are moving
the type
variable $X$ from the typing of a channel in the conclusion---$x: X
			\parr
			\assume{\query\Gamma, l:X}$---to
the typing of a process in the premise---$app: ( \query\Gamma, l:X )$. This is then carried 
over to the application of \rname{Id}, which is thus able to type the usage of a process 
variable that is generic on the behaviour that will be enacted.

\subsection{Semantics}
\label{sec:semantics}

\newcommand{\names}{\m{n}}
\newcommand{\dom}{\m{dom}}

To give a semantics to \CHOP, we follow the approach that we outlined in \cref{sec:from_cp_to_chop}:
we derive term reductions (denoted $\reducesto$) and equivalences (denoted $\equiv$) from sound proof
transformations. We distinguish between principal reductions, which reduce cuts and chops that act
directly on their premises (e.g., compatible output and input terms), and reductions that arise
from commuting conversions, which just ``push'' cuts or chops up a proof.
The full proofs of principal reductions related to process mobility (the hallmark of \CHOP) are
displayed
in  \cref{fig:proof_red}, and a set of representative commuting conversions are displayed in
\cref{fig:proof_equiv} (the dashed lines are just visual separators between different commuting
conversions, for readability). The reader familiar with linear logic will recognise that principal reductions follow the same methodology of principal cut reductions in \CLL, but
now extended to our new cases and also applied to the principal cases that we get for our new rule \rname{Chop}.
\begin{figure}
\[
\begin{array}{r@{\!}c@{\!}l}
\begin{array}{l}
	\infer[\rname{Cut}]{
		\judge{\Theta,\Theta'}{\cpres x{\provide{\Delta}}y{} ( \sendho x{\abstr \rho P} \pp \recv 
xpQ
)}{\Gamma}
	}{
		\infer[\rname{\provide{}}]{
			\judge{\Theta}{\sendho x{\abstr \rho P}}{x:\provide{\Delta}}
		}{\judge{\Theta}{P}{\Delta\rho}}
	&
	\infer[\rname{\assume{}}]
	{
		\judge{\Theta'}{\recv xpQ}{\Gamma, y:\assume{\Delta}}
	}{
		\judge{\Theta', p: \Delta}{Q}{\Gamma}
	}
	}
\end{array}
&\reducesto&
\begin{array}{c}
	\infer[\rname{Chop}]{
		\judge{\Theta,\Theta'}{Q \letp p{\abstr \rho P}}{\Gamma}
	}{
		\judge{\Theta}{P}{\Delta\rho}
		&
		\judge{\Theta',p:\Delta}{Q}{\Gamma}
	}
\end{array}
\\\\
\begin{array}{c}
	\infer[\rname{Chop}]{
		\judge{\Theta}{\invoke p{\rho} \letp p{\abstr {\rho'}P}}{\Gamma\rho}
	}{
		\judge{\Theta}{P}{\Gamma\rho'}
		&
		\infer[\rname{Id}]
		{\judge{p:\Gamma}{\invoke p {\rho}}{\Gamma\rho}}{}
	}
\end{array}
&\reducesto&
\begin{array}{c}
\judge{\Theta}{P\{\rho \circ {\rho'}^{-1}\}}{\Gamma\rho}
\end{array}
\end{array}
\]
\caption{\CHOP, Principal Proof Reductions for Process Mobility.}
\label{fig:proof_red}
\end{figure}
\begin{figure}
\[
\begin{array}{l}
\begin{array}{l}
	\infer[\rname{Chop}]{
		\judge{\Theta,\Theta',\Theta''}{\left( \cpres x{A}y{} ( P \pp Q ) \right) \letp r{\abstr \rho R}}{\Gamma,\Gamma'}
	}{
		\judge{\Theta}{R}{\Delta\rho}
	&
		\infer[\rname{Cut}]
		{\judge{\Theta,r:\Delta}{\cpres x{A}y{} ( P \pp Q )}{\Gamma,\Gamma'}}
		{
			\judge{\Theta',r:\Delta}{P}{\Gamma, x:A}
			&
			\judge{\Theta''}{Q}{\Gamma',y:\dual A}
		}
	}
\end{array}
\reducesto
\\\\
\hspace{3cm}
\begin{array}{c}
	\infer[\rname{Cut}]
		{\judge{\Theta,\Theta',\Theta''}{\cpres x{A}y{} ( P\letp r{\abstr \rho R} \pp Q )}{\Gamma,\Gamma'}}
		{
			\infer[\rname{Chop}]{
				\judge{\Theta,\Theta'}{P \letp r{\abstr \rho R}}{\Gamma, x:A}
			}{
				\judge{\Theta}{R}{\Delta\rho}
				&
				\judge{\Theta',r:\Delta}{P}{\Gamma, x:A}
			}
			&
			\judge{\Theta''}{Q}{\Gamma',y:\dual A}
		}
\end{array}
\\\\
\arrayrulecolor{gray}
\hdashline
\\
\begin{array}{c}
\infer[\rname{Chop}]
{
	\judge{\Theta,\Theta',\Theta''}{\left( \send xyPQ \right) \letp r{\abstr \rho R}}{\Gamma,\Gamma',x:A\tensor B}
}
{
	\judge{\Theta}{R}{\Delta\rho}
	&
	\infer[\tensor]
	{
		\judge{\Theta',\Theta'',r:\Delta}{\send xyPQ}{\Gamma,\Gamma',x:A\tensor B}
	}
	{
		\judge{\Theta',r:\Delta}{P}{\Gamma,y:A}
		&
		\judge{\Theta''}{Q}{\Gamma',x:B}
	}
}
\end{array}
\reducesto
\\\\
\hfill
\begin{array}{c}
\infer[\tensor]
{
	\judge{\Theta,\Theta',\Theta''}{\send xy{P \letp r{\abstr \rho R}}{Q}}{\Gamma,\Gamma',x:A\tensor B}
}
{
	\infer[\rname{Chop}]
	{\judge{\Theta,\Theta'}{P\letp r{\abstr \rho R}}{\Gamma,y:A}}
	{\judge{\Theta}{R}{\Delta\rho} & \judge{\Theta',r:\Delta}{P}{\Gamma,y:A}}
	&
	\judge{\Theta''}{Q}{\Gamma',x:B}
}
\end{array}
\\\\
\hdashline
\\
\begin{array}{c}
\infer[\rname{Chop}]
{
	\judge{\Theta,\Theta'}{\left( \recv xyP \right) \letp r{\abstr \rho R}}{\Gamma, x:A \parr B}
}
{
	\judge{\Theta}{R}{\Delta\rho}
	&
	\infer[\parr]
	{
		\judge{\Theta', r:\Delta}{\recv xyP}{\Gamma, x:A \parr B}
	}
	{
		\judge{\Theta', r:\Delta}{P}{\Gamma, y:A, x:B}
	}
}
\end{array}
\reducesto
\\\\
\hfill
\begin{array}{c}
\infer[\parr]{\judge{\Theta,\Theta'}{\recv xy{\left( P \letp r{\abstr \rho R} \right)}}{\Gamma, x:A\parr B}}
{
	\infer[\rname{Chop}]{
		\judge{\Theta,\Theta'}{P \letp r{\abstr \rho R}}{\Gamma, x:A\parr B}
	}{
		\judge{\Theta}{R}{\Delta\rho}
		&
		\judge{\Theta', r:\Delta}{P}{\Gamma, y:A, x:B}
	}
}
\end{array}
\end{array}
\]
\caption{\CHOP, Representative Commuting Conversions for Explicit Substitutions.}
\label{fig:proof_equiv}
\end{figure}

The first reduction in \cref{fig:proof_red} shows how a process communication is reduced to an explicit
substitution of the transmitted process at the receiver. The second
reduction shows how an explicit substitution replaces a process invocation. The reductions in \cref{fig:proof_equiv} illustrate how explicit substitutions float inside of terms.

The full sets of reductions and equivalences supported by \CHOP are displayed in \cref{fig:chop_red} (principal reductions), 
\cref{fig:chop_axiom_red} ($\eta$-expansion), and \cref{fig:chop_equiv} (commuting conversions and
structural equivalences).
We use the following notation: $\fn(P)$ is the set of free
names (channels and process variables) in $P$; $\names(\Delta)$ is the set of names (actually 
labels, in the case of the $\eta$-expansion in \cref{fig:chop_axiom_red}) associated to types in 
$\Delta$; and $\dom(\rho)$ is the preimage of $\rho$ (the defined 
domain of $\rho$, interpreting it as a partial function, i.e., the set of labels associated to 
values in the record $\rho$).

We use a simplified form of the second reduction proven in \cref{fig:proof_red}: 
$\invoke p \rho \letp p{\abstr \rho P} \reducesto P$. This is obtained by assuming that 
$\alpha$-conversion is used on the abstraction $\abstr \rho P$ until all the names bound by the 
abstraction term ($\rho$) match exactly the names used by the invocation ($\invoke p \rho$).

An alternative to having $\eta$-expansion is to generalise our reduction
$\cpres{x}{X}{y}{}( \cpaxiom{w}{}{x}{X} \pp P ) \reducesto P \{ w/y \}$, which is restricted to 
atomic propositions $X$, to the form $\cpres{x}{A}{y}{}( \cpaxiom{w}{}{x}{A} \pp P ) \reducesto 
P\{w/y\}$, which is not. Both are sound in \CHOP, and the choice on which one should be adopted thus 
depends on other criteria. Here we choose to use the formulation based on $\eta$-expansion, because 
it will be useful later in our extension to multiparty sessions (\cref{sec:mchop}).

\begin{figure}
\begin{equationsl}
\phl{\cpres{x}{A\tensor B}{y}{\dual A\parr \dual B}
(\send x{x'}PQ  \pp \recv y{y'}R)}
&\reducesto &
\phl{\cpres{x'}{A}{y'}{\dual A}
(P  \pp \cpres xBy{\dual B} (Q\pp R))}
\\
\cpres{x}{A\oplus B}{y}{}
( \inl xP \pp \piCase{y}{Q}{R} )
&\reducesto&
\cpres{x}{A}{y}{}
( P \pp Q )
\\
\cpres{x}{A\oplus B}{y}{}
( \inr xP \pp \piCase{y}{Q}{R} )
&\reducesto&
\cpres{x}{A}{y}{}
( P \pp R )
\\
\phl{\cpres x{\query A}y{\bang A}(\client xuQ \pp \srv yvP)}
&\reducesto&
\phl{\cpres u{A}v{\dual A}(P \pp Q)}
\\
\phl{\cpres x{\query A}y{\bang A}(P \pp \srv yvQ)}
&\reducesto&
\phl{P}
\\
\phl{\cpres x{\query A}y{\bang A}(P \substtwo x{x'}x{x''} \pp \srv yvQ)}
&\reducesto&
\multicolumn{2}{l}{
\phl{\cpres {x'}{\query A}{y'}{\bang A}
\big( \cpres {x''}{\query A}{y''}{\bang A}
(P \pp \srv {y''}vQ) \pp \srv{y'}vQ \big)}
}
\\
\phl{\cpres x{\exists X.B}y{\forall X.\dual B}(\sendtype xA{P\pp \recvtype yXQ})}
&\reducesto&
\phl{\cpres x{B \substone AX}y{B \substone AX}(P \pp Q \substone AX)}
\\
& \makebox[0em][c]{\text{(no rule for $\zero$ with $\top$)}} &
\\
\phl{\cpres x{\one}y\perp(\close x\pp \wait yP)}
&\reducesto&
\phl{P}
\\
\cpres x{\provide{\Sigma;\Delta}}y{} ( \sendho x{\abstr \rho P} \pp \recv xpQ )
&\reducesto&
Q \letp p{\abstr \rho P}
\\
\invoke p \rho \letp p{\abstr \rho P} & \reducesto & P
\\
\cpres{x}{X}{y}{}( \cpaxiom{w}{}{x}{X} \pp P )
&\reducesto&
P \{ w/y \}
\end{equationsl}
\caption{\CHOP, Principal Reductions.}
\label{fig:chop_red}
\end{figure}
\begin{figure}
\[\begin{array}{rcl@{\quad}rcl}
\cpaxiom{x}{\dual A \parr \dual B}{y}{A \tensor B}
      &\reducesto&
      \phl{\recv xu{\send yv{\cpaxiom{u}{\dual A}{v}{A}}{\cpaxiom{x}{\dual B}{y}{B}}}} &
      \phl{\cpaxiom{x}{\perp}{y}{\one}}
      &\reducesto&
      \phl{\wait{x}{\close{y}}} \\
      \phl{\cpaxiom{x}{\dual A \with \dual B}{y}{A \oplus B}}
      &\reducesto&
      \phl{\piCase x{\inl y{\cpaxiom{x}{\dual A}{y}{A}}}{\inr y {\cpaxiom{x}{\dual B}{y}{B}}}} &
      \phl{\cpaxiom{x}{\top}{y}{\zero}}
      &\reducesto& \phl{\emptyCase x} \\
      \phl{\cpaxiom{x}{\bang \dual A}{y}{\query A}}
      &\reducesto&
      \phl{\srv xu{\client yv{\cpaxiom{u}{\dual A}{v}{A}}}} &
      \phl{\cpaxiom{x}{\forall X.\dual A}{y}{\exists X.A}}
      &\reducesto&
      \phl{\recvtype xX{\sendtype yX{\cpaxiom{x}{\dual A}{y}{A}}}}
      \\
      \cpaxiom{x}{\assume{\Sigma;\Delta}}{y}{\provide{\Sigma;\Delta}}
      &\reducesto&
      \recv xp{\sendho y{\abstr{\rho}{\invoke p\rho}}}\quad \mbox{ if } \dom(\rho) = \names(\Sigma) \cup \names(\Delta)
\end{array}\]
\caption{\CHOP, $\eta$-expansions.}
\label{fig:chop_axiom_red}
\end{figure}
\begin{figure}
\begin{equationsl}
\cpaxiom{x}{}{y}{A} &\equiv& \cpaxiom{y}{}{x}{\dual A}
\\
\cpres{z}{A}{w}{}(P \pp Q)
&\equiv&
\cpres{w}{\dual A}{z}{}(Q\pp P)
\\
\cpres{z}{A}{w}{}\left( \cpres{x}{B}{y}{} ( P \pp Q ) \pp R \right)
&\equiv&
\cpres{x}{B}{y}{}\left( P \pp \cpres{z}{A}{w}{}( Q \pp R ) \right)
\\
\cpres{z}{A}{w}{}(\send xyPQ \pp R)
&\reducesto&
\send xy{\cpres{z}{A}{w}{}(P \pp R)}{Q}
&
\text{if } z \in \fn(P)
\\
\cpres{z}{A}{w}{}(\send xyPQ \pp R)
&\reducesto&
\send xyP{\cpres{z}{A}{w}{}(Q \pp R)}q
&
\text{if } z \in \fn(Q)
\\
\cpres{z}{A}{w}{}(\recv xyP \pp Q)
&\reducesto&
\recvs xy. \cpres{z}{A}{w}{}(P \pp Q)
\\
\cpres{z}{A}{w}{}(\inl xP \pp Q)
&\reducesto&
\inl x{\cpres{z}{A}{w}{}(P \pp Q)}
\\
\cpres{z}{A}{w}{}(\inr xP \pp Q)
&\reducesto&
\inr x{\cpres{z}{A}{w}{}(P \pp Q)}
\\
\cpres{z}{A}{w}{}( \piCase xPQ \pp R )
&\reducesto&
\piCase x{\cpres{z}{A}{w}{}( P \pp R )}{\cpres{z}{A}{w}{}( Q \pp R )}
\\
\cpres{z}{A}{w}{}(\srv xyP \pp Q)
&\reducesto&
\srv xy{\cpres{z}{A}{w}{}(P \pp Q)}
\\
\cpres{z}{A}{w}{}(\client xyP \pp Q)
&\reducesto&
\client xy{\cpres{z}{A}{w}{}(P \pp Q)}
\\
\cpres{z}{A}{w}{}(\sendtype xBP \pp Q)
&\reducesto&
\sendtype xB{\cpres{z}{A}{w}{}(P \pp Q)}
\\
\cpres{z}{A}{w}{}(\recvtype xXP \pp Q)
&\reducesto&
\recvtype xX{\cpres{z}{A}{w}{}(P \pp Q)}
\\
\cpres{z}{A}{w}{}(\wait xP \pp Q)
&\reducesto&
\wait x{\cpres{z}{A}{w}{}(P \pp Q)}
\\
\cpres{z}{A}{w}{}( \emptyCase x \pp Q )
&\reducesto&
\emptyCase x
\\
\left( \cpres x{A}y{} ( P \pp Q ) \right) \letp r{\abstr \rho R}
&\reducesto&
\cpres x{A}y{} ( P \letp r{\abstr \rho R} \pp Q )
& \mbox{if } r \in\fn(P)
\\
\left( \cpres x{A}y{} ( P \pp Q ) \right) \letp r{\abstr \rho R}
&\reducesto&
\cpres x{A}y{} ( P \pp Q \letp r{\abstr \rho R} ) & \mbox{if } r \in \fn(Q)
\\
\left( \send xyPQ \right) \letp r{\abstr \rho R}
&\reducesto&
\send xy{P \letp rR}{Q}
& \mbox{if } r \in \fn(P)
\\
\left( \send xyPQ \right) \letp r{\abstr \rho R}
&\reducesto&
\send xy{P}{Q \letp r{\abstr \rho R}}
& \mbox{if } r \in \fn(Q)
\\
\left( \recv xyP \right) \letp r{\abstr \rho R}
&\reducesto&
\recv xy{\left( P\letp r{\abstr \rho R} \right)}
\\
(\inl xP)\letp r{\abstr \rho R}
&\reducesto&
\inl x{(P \letp r{\abstr \rho R})}
\\
(\inr xP)\letp r{\abstr \rho R}
&\reducesto&
\inr x{(P \letp r{\abstr \rho R})}
\\
\left( \piCase x{P}{Q} \right) \letp r{\abstr \rho R}
&\reducesto&
\piCase x{P \letp r{\abstr \rho R}}{Q \letp r{\abstr \rho R}}
\\
\left( \client xyP \right)\letp r{\abstr \rho R}
&\reducesto&
\client xy{(P \letp r{\abstr \rho R})}
\\
\left( \sendtype xAP \right)\letp r{\abstr \rho R}
&\reducesto&
\sendtype xA{(P \letp r{\abstr \rho R})}
\\
\left( \recvtype xXP \right)\letp r{\abstr \rho R}
&\reducesto&
\recvtype xX{(P \letp r{\abstr \rho R})}
\\
\left( \sendho x{\abstr {\rho'} P} \right)\letp r{\abstr \rho R}
&\reducesto&
\sendho x{\abstr {\rho'}{(P \letp r{\abstr \rho R})}}
\\
\left( \recv xpP \right)\letp r{\abstr \rho R}
&\reducesto&
\recv xp{(P \letp r{\abstr \rho R})}
\\
\left( \wait xP \right)\letp r{\abstr \rho R}
&\reducesto&
\wait x{(P \letp r{\abstr \rho R})}
\\
\multicolumn{3}{c}{\text{(no rule for $\zero$)}}
\end{equationsl}
\caption{\CHOP, Equivalences and Commuting Conversions.}
\label{fig:chop_equiv}
\end{figure}

In the remainder, we allow reductions to silently apply equivalences. We also
allow reductions to happen inside of cuts and chops, as formalised in the following.
\[
\begin{array}{rcl@{\qquad}l}
	\cpres{x}{A}{y}{}(P_1 \pp Q)
    & \reducesto & \cpres{x}{A}{y}{}(P_2 \pp Q)
    & \text{if } P_1 \reducesto P_2
    \\
    \cpres{x}{A}{y}{}(P \pp Q_1)
    & \reducesto & \cpres{x}{A}{y}{}(P \pp Q_2)
    & \text{if } Q_1 \reducesto Q_2
    \\
    Q_1 \letp p{\abstr \rho P}
    &\reducesto&
    Q_2 \letp p{\abstr \rho P}
    & \text{if } Q_1 \reducesto Q_2
\end{array}
\]

\subsection{Example, semantics}


Take the improved cloud server given in \cref{sec:cloud_server}, which we write here again for
convenience and call $P_{srv}$.
\[
\begin{array}{l@{~}l@{\quad}l}
P_{srv} \defeq &
\srv{cs}{x}{\ x(X).\ \piCase{x}{
&
	\recvs{x}{x'}. \recvs {x}{app}. \recvs {x'}{db}. \cpres{z}{X}{w}{}( \invoke{app}\rho \pp \invoke{db}{\rho'} )
}{
\\ &&
	\recvs{x}{extdb}. \recvs{x}{app}. \cpres{z}{X}{w}{}( \invoke {app}{\rho} \pp \cpaxiom{extdb}{}{w}
	{\dual{X}} ) \quad } }
\end{array}
\]

Here is a simple compatible client, $P_{cli}$. It uses the second option: connecting the
application in
the cloud server to an existing database. We run this database in the client itself ($P_{db}$).
This
makes sense in many practical scenarios, where the client wants to maintain ownership of its data
but also
to delegate heavy computational tasks to the cloud. Type $A$ is the protocol, left unspecified,
between
the application sent by the client ($P_{app}$) and its database.
\[
P_{cli} \defeq
\client {cc}{y}{
	\sendtype yA{
		\inr y{
			\send y{d}{
				P_{db}
			}{
				\sendho y{\abstr {\rho} P_{app}}
			}
		}
	}
}
\]

Process $P_{cli}$ implements on channel $cc$ (for cloud client) the dual type of that implemented
by
$P_{srv}$ on channel $cs$ (for cloud server), discussed in \cref{sec:cloud_server}. So we can
compose them
in parallel (we omit types at restrictions, for brevity).
\[
P_{sys} \defeq \cpres{cc}{}{~cs}{}( P_{cli} \pp P_{srv} )
\]
The semantics of our system $P_{sys}$ depends on whether $P_{db}$ uses channel $cc$ (to use the
cloud
server itself) or not. If it does, then we have to replicate $P_{srv}$, otherwise
we have to execute a single instance of it. Let us assume that we need a single instance ($P_{db}$
does
not use $cc$). Then we have the following reduction chain for the first three communications
(invocation
of the server, communication of the protocol $A$, right selection).
\[
P_{sys} \reducesto\reducesto\reducesto
\cpres{y}{}{x}{}
(
\send y{d}{P_{db}}{\sendho y{\abstr{\rho}{P_{app}}}}
\pp
\recvs{x}{extdb}. \recvs{x}{app}. \cpres{z}{A}{w}{}( \invoke{app}{\rho} \pp \cpaxiom{extdb}{}{w}{\dual{A}} )
)
\]
Notice that the server now knows the protocol $A$ between the application and the database.
A further reduction gives us the following process.
\[
\reducesto \cpres{d}{}{~extdb}{}\left(
	P_{db}
	\pp
	\cpres{y}{}{x}{}\left( \sendho y{\abstr \rho {P_{app}}} \pp
	\recvs{x}{app}. \cpres{z}{A}{w}{}( \invoke{app}{\rho} \pp \cpaxiom{extdb}{}{w}{\dual{A}} )
	\right)
\right)
\]
We can now reduce the communication between $y$ and $x$, so that now $P_{app}$ is moved from the
client
to server, and reduce the resulting explicit substitution at the server, as follows.
\[
\reducesto \reducesto \cpres{d}{}{~extdb}{}\left(
	P_{db}
	\pp
	\cpres{z}{A}{w}{}( P_{app} \pp \cpaxiom{extdb}{}{w}{\dual{A}} )
\right)
\]
So we now have that the database of the client communicates with the application in the server, as we originally intended.

\section{Metatheory}
\label{sec:meta}

We now move to the metatheoretical properties of \CHOP.

\subsection{Type preservation}
Since all reductions and equivalences in \CHOP are derived from type-preserving proof
transformations, we
immediately get a type preservation result for both.
\begin{theorem}[Type preservation]
\label{thm:type_preservation}
Let $\judge{\Theta}{P}{\Gamma}$. Then,
\begin{itemize}
\item $P\equiv Q$ for some $Q$ implies $\judge{\Theta}{Q}{\Gamma}$, and
\item $P \reducesto Q$ for some $Q$ implies $\judge{\Theta}{Q}{\Gamma}$.
\end{itemize}
\end{theorem}
\begin{proof}
See \cref{fig:proof_red} and \cref{fig:proof_equiv} for the most interesting cases. The other cases
are
similar.
\end{proof}

\newcommand{\tostar}{\reducesto^*}

\subsection{Progress for cuts}
An important property of \CP is that all well-typed processes can make communications progress, in
the sense that all top-level cuts can eventually be eliminated by applying reductions
and equivalences. In \CHOP, we need to be more careful when studying progress. To see why,
consider the following process.
\[
\judge{p:(l:\one)}{\cpres{x}{\one}{y}{}( \invoke{p}{l=x} \pp \wait y{\close z} )}{z:\one}
\]
The process above is well-typed but stuck (cannot progress), because of the free process
variable $p$. In other words, we lack the code for the left participant in the cut.
Notice that here we are not stuck because of incompatible communication structures between
participants,
but rather because we still do not know the entirety of the code that we are supposed to execute.

In general, in \CHOP, we may be in a situation where some code is missing. But if we are not
missing any
code, then we should not have problems in executing all communications.
We call a process open if it has some free process variables, and closed if it has no free process
variables. We shall also write that a process $P$ is a
cut if its proof ends with an application of \rname{Cut} (so $P$ is a restriction term), and that
$P$ is a chop if its proof ends with
an application of \rname{Chop} (so $P$ is an explicit substitution term).
Then, \CHOP enjoys progress in the sense that any closed process that is a restriction can be
reduced.
To prove this, we first establish the following lemma, which states that chops cannot get us 
stuck.
\begin{lemma}[Progress for top-level chops]
\label{lemma:chop_progress}
Let $\judge{\Theta}{P}{\Gamma}$ be a chop. Then, there exists $Q$ 
such that $P \reducesto Q$ and $\judge{\Theta}{Q}{\Gamma}$.
\end{lemma}
\begin{proof}
Because of type preservation (\cref{thm:type_preservation}), we only need to find a $Q$
such that $P \reducesto Q$.

We proceed by induction on the structure of $P$. The cases are on the last applied rule for the 
right premise of the chop.

If the right premise is itself a chop, we invoke the induction hypothesis and lift the 
corresponding reduction.

For all other cases, we can either eliminate the top-level chop immediately with a reduction from
\cref{fig:chop_red} or push the chop up the proof with a reduction from \cref{fig:chop_equiv}.
\end{proof}

Recall that we write $\judge{}{P}{\Gamma}$ when the process
environment in the judgement is empty.

\begin{theorem}[Progress for top-level cuts]
\label{thm:cut_progress}
Let $\judge{\ }{P}{\Gamma}$ and $P$ be a cut.
Then, there exists $Q$ such that $P \reducesto Q$ and $\judge{\ }{Q}{\Gamma}$.
\end{theorem}
\begin{proof}
Because of type preservation (\cref{thm:type_preservation}), we only need to find a $Q$
such that $P \reducesto Q$.

We now proceed by cases on the last applied rules of the premises of the cut.

If one of the premises is itself a cut, we recursively eliminate it (apply a reduction and lift it 
to the top-level cut).

If the premises are both applications of rules that act on the channels used in the
restriction, then we can apply one of the reductions in \cref{fig:chop_red}. (This may require 
using an equivalence.)

If one of the premises acts on a channel not named in the
restriction, then we apply one of the commuting conversions in \cref{fig:chop_equiv}.

If one of the premises is a chop, we apply Lemma~\ref{lemma:chop_progress} and lift the reduction 
to the cut.
\end{proof}

\subsection{General progress and execution strategy}

Establishing progress only for top-level cuts is not entirely satisfactory, because it does not
give us an execution strategy for \CHOP. What about processes with top-level chops? (Meaning
that an explicit substitution is at the top level.) For
example, consider the following term, where $\rho = \{l_1=x,l_2=y,l_3=z\}$.
\[
\invoke p{\rho} \letp{p}{\abstr{\rho}{\ \cpres{x}{\one}{y}{}( \close x \pp \wait y{\close z} )}}
\]
This process does not have any communications to perform at the top level, so if we just consider
top-level cut elimination we would not execute anything here.
However, intuitively we would expect this process to run $p$ by using the explicit substitution and
reduce as follows.
\[
\invoke p{\rho} \letp{p}{\abstr{\rho}{\ \cpres{x}{\one}{y}{}( \close x \pp \wait y{\close z} )}}
\reducesto \reducesto \close z
\]

So we need a more general progress result, which considers also chops.

\begin{theorem}[Progress]
\label{thm:progress}
Let $\judge{\ }{P}{\Gamma}$ be a cut or a chop. Then, there exists $Q$ such that $P \reducesto Q$
and
$\judge{\ }{Q}{\Gamma}$.
\end{theorem}
\begin{proof}
Immediate consequence of \cref{thm:cut_progress} and \cref{lemma:chop_progress}.
\end{proof}

\begin{remark}
The proof of progress gives us an execution strategy that reaps the benefits of using explicit substitutions.
Consider the following example (we omit the type of the restriction).
\[
\cpres{x}{}{y}{}
\left(
\send{x}{z}{
	\sendho{z}{\abstr \rho {P_{BIG}}}
}{
	\inl{x}{P}
}
\pp
\recvs y{w}.\recv wp{
	\ \piCase{y}{Q}{R}
}
\right)
\]
This can reduce as follows.
\[
\reducesto\reducesto
\cpres{x}{}{y}{}
\left(
\inl{x}{P}
\pp
\piCase{y}{Q\letp p{\ \abstr \rho {P_{BIG}}}}{R \letp p{\ \abstr \rho {P_{BIG}}}}
\right)
\]
Assume now that $P_{BIG}$ is some big process term, and that both $Q$ and $R$ use variable $p$.
The duplication of the explicit substitution that we have obtained with the reduction chain is
cheap in practice, since we could implement it by using a shared pointer to a representation of $P_{BIG}$.
This is safe, because the processes inside of explicit substitutions are ``frozen'' (cannot be
reduced).

However, if we chose to equip \CHOP with a semantics that resolves substitutions immediately, as in standard \HOpi, we would put the code of
$P_{BIG}$ in both $Q$ and $R$, leading to a costly duplication.
Instead, the execution strategy in the proof of \cref{thm:progress} is lazy with respect to
substitutions and prefers the following reduction.
\[
\reducesto
\cpres{x}{}{y}{}
\left(
P
\pp
Q\letp p{\ \abstr \rho {P_{BIG}}}
\right)
\]
In the term above, there is no risk of duplicating the code of $P_{BIG}$ unnecessarily anymore.
\end{remark}

\subsection{Chop elimination}

Rule \rname{Chop} can be seen as a new rule for composing processes, similar to rule \rname{Cut},
but acting on process variables instead of connecting channels.
Since linear logic supports cut elimination (all cuts can be removed from proofs), it is natural to
ask whether the same can be done for chops.
We prove here that \CHOP supports chop elimination. We say that $P$ is chop-free if its proof does
not contain any applications of rule \rname{Chop}.
Eliminating all chops from a proof corresponds to replacing all invocations of a process variable 
with the body of the process
specified by the corresponding chop that instantiates the variable.
This means that \CHOP can be seen as a logical reconstruction of \HOpi, where all substitustions 
are performed immediately after processes are communicated: explicit substitutions just give an 
implementation strategy for this mechanism.

We first establish that rule \rname{Chop} is admissible.
\begin{theorem}[Chop admissibility]
\label{thm:chop_adm}
Let $\judge{\Theta}{P}{\Delta\rho}$ and $\judge{\Theta',p:\Delta}{Q}{\Gamma}$, 
where $P$ and $Q$ are chop-free.
Then, there exists $R$ such that $\judge{\Theta,\Theta'}{R}{\Gamma}$ and $R$ is chop-free.
\end{theorem}
\begin{proof}
By induction on the structure of the proof of $\judge{\Theta',p:\Delta}{Q}{\Gamma}$. We 
proceed by cases on the last applied rule in this proof.

If the last applied rule is \rname{Id}, then $Q=\invoke p{\rho'}$ for some $\rho'$, 
$\Theta' = \cdot$, and $\Gamma = \Delta\rho'$.
Then $R = P \{ \rho' \circ \rho^{-1} \}$ and the thesis follows (using appropriate 
$\alpha$-renaming 
to avoid clashes between the free names used in $\rho'$ and those in $P$).

For all other cases, we proceed as for the corresponding commuting conversion in
\cref{fig:chop_equiv} (interpreting chop as an admissible rule) and then invoke the induction
hypothesis.
\end{proof}

The proof of chop admissibility gives us a methodology for eliminating all chops from proofs.
\begin{theorem}[Chop elimination]
\label{thm:chop_elim}
Let $\judge{\Theta}{P}{\Gamma}$. Then, there exists $Q$ such that $Q$ is chop-free and
$\judge{\Theta}{Q}{\Gamma}$.
\end{theorem}
\begin{proof}
We iteratively eliminate the inner-most applications of rule \rname{Chop}
in the proof of $P$, until there are no more chops. To eliminate these chops, we follow the
inductive procedure described in the proof of \cref{thm:chop_adm}, since inner-most applications
have chop-free premises.
\end{proof}

\section{Extensions}
\label{sec:extensions}

We illustrate how the basic theory of \CHOP can be used to obtain richer features.

\subsection{Derivable constructs as syntax sugar}
\label{sec:proc_output_wc}

We derive some constructs using the proof theory of \CHOP and provide them as syntactic sugar.

\paragraph{Output of free channel names}

In \CHOP, we can only output bound channel names, as in the internal $\pi$-calculus by \citet{S96}.
It is well-known that link terms can be used to simulate the output of a free channel. In \CP, this 
works, as shown by \citet{LM15}. The same applies to \CHOP. We can add the usual term for free 
channel output to the syntax of processes.

\begin{displaymath}
\begin{array}{rll}
P,Q,R ::= & \cdots & \\
& \sendfree{x}{y}{P} & \text{output a free channel name and continue}
\end{array}
\end{displaymath}

This new term is desugared as follows.

\[
\begin{array}{rcl}
\sendfree xyP &\defeq &
\send xz{ \cpaxiom{y}{A}{z}{\dual A} }{ P }
\end{array}
\]

The proof of soundness is easy, but reading it is useful to understand the typing of the new 
construct.

\begin{displaymath}
\begin{array}{c}
\infer[]{
	\judge{\Theta}{\sendfree xyP}{\Gamma, y:\dual A, x: A 
\tensor B}
}{
	\judge{\Theta}{P}{\Gamma, x: B}
}
\end{array}
\quad \defeq \quad
\begin{array}{c}
\infer[\tensor]{
	\judge{\Theta}{\send xz{ \cpaxiom{y}{A}{z}{\dual A} }{ P }}{\Gamma, y:\dual A, x: A 
\tensor B}
}{
	\infer[\rname{Axiom}]{
		\judge{}{\cpaxiom{y}{\dual A}{z}{A}}{y:\dual A, z: A}
	}{}
	&
	\judge{\Theta}{P}{\Gamma, x: B}
}
\end{array}
\end{displaymath}

\paragraph{Higher-order I/O with continuations}
We can also derive constructs for sending and receiving processes over channels and then allow us
to continue using that channel. We distinguish these sugared counterparts from our output and input
primitives by using bold brackets.
\begin{displaymath}
\begin{array}{rll}
P,Q,R ::= & \cdots & \\
& \sendhowc{x}{\abstr \rho P}{Q} & \text{output a process and continue}
\\
& \recvhowc{x}{p}{P} & \text{input a process and continue}
\end{array}
\end{displaymath}

These constructs are desugared as follows (we show directly the proofs).
\[
\begin{array}{rcl}
	\begin{array}{c}
		\infer{
			\judge{\Theta,\Theta'}{\sendhowc{x}
				{\abstr \rho P}{Q}}{\Gamma,x:\provide{\Delta}\tensor A}
		}{
			\judge{\Theta}{P}{\Delta\rho}
			&
			\judge{\Theta'}{Q}{\Gamma, x:A}
		}
	\end{array}
	&\defeq &
	\begin{array}{c}
	\infer[\tensor]{
		\judge{\Theta,\Theta'}{\send xy{\sendho y{\abstr \rho 
P}}{Q}}{\Gamma, x:\provide{\Delta}\tensor A}
	}{
		\infer[\provide{}]{
			\judge{\Theta}{\sendho y{\abstr \rho P}}{y:\provide{\Delta}}
		}{
			\judge{\Theta}{P}{\Delta\rho}
		}
		&
		\judge{\Theta'}{Q}{\Gamma, x:A}
	}
	\end{array}
\\\\
	\begin{array}{c}
		\infer{
			\judge{\Theta}{\recvhowc{x}{p}{P}}{\Gamma, x:\assume{\Delta} \parr A}
		}{
			\judge{\Theta, p:\Delta}{P}{\Gamma, x: A}
		}
	\end{array}
	&\defeq &
	\begin{array}{c}
	\infer[\parr]{
		\judge{\Theta}{\recvs xy.\recvs yp.P}{\Gamma, x:\assume{\Delta} \parr A}
	}{
		\infer[\assume{}]{
			\judge{\Theta}{\recvs yp.P}{\Gamma, y:\assume{\Delta}, x: A}
		}{
			\judge{\Theta, p:\Delta}{P}{\Gamma, x: A}
		}
	}
	\end{array}
\end{array}
\]

\paragraph{Procedures}
If the abstraction that we send over a channel does not refer to any free 
process variable, then we can always replicate it as many times as we wish. Here is the proof.
\[
\infer[\bang]{
	\judge{}{\srv xy{\sendho y{\abstr \rho P}}}
	{x:\bang \provide{\Gamma}}
}{
	\infer[\rname{\provide{}}]
      {\judge{}{\sendho y{\abstr \rho P}}{y: \provide{\Gamma}}}
      {\judge{}{P}{\Gamma\rho}}
}
\]
We use this property to build a notion of procedures that can be used at will. We denote procedure 
names with $K$, for readability. We will later use it as a channel name in our desugaring.
\begin{displaymath}
\begin{array}{rll}
P,Q,R ::= & \cdots & \\
& \m{def} \; K = \abstr{\rho}{P} \; \m{in}\; Q & \text{procedure definition}
\\
& \invoke K{\rho} & \text{procedure invocation}
\end{array}
\end{displaymath}
A term $\m{def} \; K = \abstr{\rho}{P} \; \m{in}\; Q$ defines procedure $K$ as $\abstr{\rho}{P}$ in 
the scope of $Q$, and a term $\invoke K{\rho}$ invokes procedure $K$ by passing the parameters 
$\rho$.

Here is the desugaring of both constructs. For simplicity of presentation, in term $\m{def} \; K = 
\abstr{\rho}{P} \; \m{in}\; Q$ we assume that $K$ is used at least once in $Q$. The generalisation 
to the case where $Q$ does not use $K$ at all is straightforward (thanks to rule \rname{Weaken}).

\[
\begin{array}{rcl}
	\begin{array}{c}
		\infer{
			\judge{\Theta}{\m{def} \; K = \abstr{\rho}{P} \; \m{in}\; Q}
				{\Gamma}
		}{
			\judge{}{P}{\Gamma\rho}
			&
			\judge{\Theta}{Q}{\Gamma, K:\query\assume{\Gamma}}
		}
	\end{array}
	&\defeq &
	\begin{array}{c}
	\infer[\rname{Cut}]{
		\judge{\Theta}{
			\cpres{x}{\bang \provide{\Gamma}}{K}{}\left(
			\srv xy{\sendho y{\abstr \rho P}}
			\pp
			Q
			\right)
		}{\Gamma}
	}{
		\infer[\bang]{
			\judge{}{\srv xy{\sendho y{\abstr \rho P}}}
			{x:\bang \provide{\Gamma}}
		}{
			\infer[\rname{\provide{}}]
			{\judge{}{\sendho y{\abstr \rho P}}{y: \provide{\Gamma}}}
			{\judge{}{P}{\Gamma\rho}}
		}
		&
		\judge{\Theta}{Q}{\Gamma, K:\query\assume{\Gamma}}
	}
	\end{array}
\\\\
	\begin{array}{c}
		\infer{
			\judge{}{\invoke K{\rho}}
			{\Gamma\rho,K:\query\assume{\Gamma}}
		}{}
	\end{array}
	&\defeq &
	\begin{array}{c}
	\infer[\query]{
		\judge{}{\client Ky{\recvs yp.\invoke p\rho}}
			{\Gamma\rho,K:\query\assume{\Gamma}}
	}{
		\infer[\assume{}]{
			\judge{}{\recvs Kp.\invoke p\rho}
			{\Gamma\rho,K:\assume{\Gamma}}
		}{
			\infer[\rname{Id}]{
				\judge{p:\Gamma}{\invoke p\rho}{\Gamma\rho}
			}{}
		}
	}
	\end{array}
\end{array}
\]

Observe that even if procedures can be used at will, typing ensures that each usage respects 
linearity (i.e., every usage ``consumes'' the necessary linear resources available in the context). Note also that self-invocations are not supported, as typing forbids them.

\paragraph{Higher-order parameters}
We chose not to make abstractions parametric on process variables for economy of the calculus.
This feature can be reconstructed with the following syntactic sugar.

\begin{displaymath}
\begin{array}{rll}
P,Q,R ::= & \cdots & \\
& x\lambda {q}.P & \text{named higher-order parameter}
\\
& \invoke P{x={\abstr{\rho} Q}} & \text{application}
\end{array}
\end{displaymath}

The desugaring is simple, interpreting a named higher-order parameter as a channel on which 
we perform a single higher-order input.

\[
\begin{array}{rcl}
	\begin{array}{c}
		\infer{
			\judge{\Theta}{x\lambda p.P}{\Gamma, x:\assume{\Delta}}
		}{
			\judge{\Theta,p:\Delta}{P}{\Gamma}
		}
	\end{array}
	&\defeq &
	\begin{array}{c}
	\infer[\assume{}]{
		\judge{\Theta}{\recvs{x}{p}.{P}}{\Gamma, x:\assume{\Delta}}
	}{
		\judge{\Theta,p:\Delta}{P}{\Gamma}
	}
	\end{array}
\\\\
	\begin{array}{c}
		\infer{
			\judge{\Theta,\Theta'}{\invoke P{x={\abstr{\rho} Q}}}{\Gamma}
		}{
			\judge{\Theta}{Q}{\Delta\rho}
			&
			\judge{\Theta'}{P}{\Gamma,x:\assume\Delta}
		}
	\end{array}
	&\defeq &
	\begin{array}{c}
	\infer[\rname{Cut}]{
		\judge{\Theta,\Theta'}{\cpres{y}{\provide{\Delta}}{x}{}}{\Gamma}
	}{
		\infer[\provide{}] {
			\judge{\Theta}{\sendho y{\abstr \rho Q}}{y:\provide{\Delta}}
		}{
			\judge{\Theta}{Q}{\Delta\rho}
		}
		&
		\judge{\Theta'}{P}{\Gamma,x:\assume\Delta}
	}
	\end{array}
\end{array}
\]

Notice that the desugaring of $\invoke P{x=\abstr \rho Q}$ yields a process that allows for 
reductions to happen in $P$ before the application takes place, since the corresponding 
higher-order named parameter term may be nested inside of $P$.
Also, this desugaring cannot be implemented by merely using an application of rule \rname{Chop}, 
since $p$ is \emph{bound} to $P$ in term $x\lambda p.P$ (as can be observed by its 
desugaring), and \rname{Chop} acts on free process variables.

\subsection{Multiparty sessions}
\label{sec:mchop}

In~\cite{CLMSW16,CMSY17}, \CP was extended to support multiparty sessions,
yielding a logical reconstruction of the theory of multiparty session types by \citet{HYC16}.
Multiparty session types allow for typing sessions with more than two
participants. A discussion of their
usefulness goes beyond the scope of this paper. The interested reader
may consult~\cite{Aetal16}. What we are interested in here
is to illustrate how \CHOP can support multiparty sessions.

The key to having multiparty sessions is to change the \rname{Cut} rule to the
following \rname{CCut} rule. We adapt the version from~\cite{CLMSW16} to our higher-order contexts. We write $\til x^{\til A}$ as an abbreviation of $x_1^{A_1}, \ldots, x_n^{A_n}$ (likewise for $\til P$, $\til \Gamma$).
\[
\infer[\rname{CCut}]
      {\phl{\gcpres{{\til x}^{\til A}}{G}{\til P}} \seq \thl{\til \Gamma}}
      {(\judge{\Theta_i}{\phl{P_i}}{\thl{\Gamma_i}, \phl{x_i}:\thl{A_i}})_i
        &
        \phl G \gseq (\phl{x_i}:\thl{{A_i}})_i}
\]
The differences between \rname{Cut} and the new \rname{CCut} (for Coherence
Cut) are: we are now composing in parallel an unbounded number of processes
$P_i$, each one implementing type $A_i$ at the respective channel $x_i$; and,
we use a new coherence judgement $\phl G \gseq (\phl{x_i}:\thl{{A_i}})_i$ to
check the compatibility of these types.
The type $G$ in the coherence judgement (and the restriction term) is a global type that prescribes
how
all participants in the session should communicate. The rules for deriving
these judgements for \CP are displayed in \cref{fig:coherence}. Coherence rules check that the
types of
participants can be matched accordingly to the global type, which is interpreted as a proof term.
\begin{figure}[t]
  \begin{displaymath}
    \begin{array}{c}
      \infer[\tensor\parr]
      {\ghl{\parrtensor{\til x}{y}{G}{H}}
        \gseq \thl \Gamma,  (\phl{x_i} : \thl{A_i \tensor B_i})_i, \phl{y}:\thl{C \parr D}}
      {\ghl G \gseq \thl (\phl{x_i}:\thl{A_i})_i, \phl y:\thl C & \ghl {H} \gseq \thl {\Gamma},
(\phl{x_i}:\thl{B_i})_i, \phl y:\thl D}
      \qquad
      \infer[\one\bot]
      {\ghl{\botone{\til x}{y}}
        \gseq (\phl {x_i}:\thl{\one})_i, \phl{y}:\thl{\bot}}
      {}
      \\[1ex]
      \infer[\oplus\with]
      {\ghl{\withplus{x}{\til y}{G}{H}}
        \gseq \thl{\Gamma}, \phl x:\thl{A \oplus B}, (\phl{y_i}:\thl{C_i \with D_i})_i}
      {  \ghl G  \gseq \phl \Gamma, \phl x:\thl A, (\phl {y_i}:\thl{C_i})_i
        & \ghl H \gseq \phl \Gamma, \phl x:\thl B, (\phl {y_i}:\thl{D_i})_i}
      \qquad
      \infer[\zero\top]
      {\ghl{\topzero{x}{\til y}}
        \gseq \thl \Gamma, \phl x:\thl \zero, (\phl{y_i}:\thl \top)_i}
      { }
      \\[1ex]
      \infer[\query\bang]
      {\ghl{\bangquery{x}{\til y}{G}}
        \gseq \phl x:\thl{\query A}, (\phl{y_i}:\thl{\bang B_i})_i}
      {\ghl G \gseq \phl x:\thl A, (\phl{y_i}:\thl{B_i})_i}
      \qquad
      \infer[\exists\forall]
      {\ghl{\forallexists{X}{x}{\til y}{G}}
        \gseq \thl \Gamma, \phl x:\thl{\exists X.A}, (\phl{y_i}:\thl{\forall X.B_i})_i}
      {  \ghl G  \gseq \thl \Gamma, \phl x:\thl A, (\phl{y_i}:\thl{B_i})_i
        & \thl X \notin \m{ftv}(\thl \Gamma)}
      \\[1ex]
      \infer[\textsc{Axiom}]
      {\ghl{\globalaxiom{x}{A}{y}{\dual{A}}} \gseq \phl x:\thl A, \phl y:\thl{\dual{A}}}
      {}
    \end{array}
  \end{displaymath}
  \caption{\CP, Coherence Rules (from~\cite{CLMSW16}).}
  \label{fig:coherence}
\end{figure}

Extending coherence to \CHOP is straightforward, as we just need to add a rule for our new type
constructors regarding process mobility. Here is the rule.
\[
\infer[\provide{}\assume{}]
{
\provideassume{x}{y}{\Delta}
\gseq x:\provide{\Delta}, y:\assume{\Delta}
}
{}
\]
And this is the principal reduction for process communication in the multiparty system.
\[
\gcpres{x^{\provide{\Delta}} y^{\assume{\Delta}}}{
	\provideassume{x}{y}{\Delta}
}{
	\sendho {x}{\abstr \rho P}
	\pp
	\recv {y}{p}{Q}
}
\reducesto
Q\letp p{\abstr \rho P}
\]
Deriving the other reductions and equivalences for \CHOP with multiparty communications is a
straightforward exercise in adapting the rules from~\cite{CLMSW16}.

Following the same idea used for building syntax sugar for process output with continuation at the
end of \cref{sec:proc_output_wc}, we could also develop syntax sugar for global types that allows
for continuations after process
communications, or even multicasts of process terms, by using auxiliary channel communications.

\section{From \CHOP to \CP}
\label{sec:translation}

\newcommand{\encho}[2]{\left\llbracket #1 \right\rrbracket_{#2}}
\newcommand{\enchopt}[2]{\encho{\ba{c}#1\ea}{#2}}

The original presentation of the $\pi$-calculus by \citet{MPW92} does not include process mobility, 
only channel
mobility. This was motivated by the belief that channel mobility can be used to
simulate the effects obtained by using process mobility. \citet{S93} later proved this belief to be 
correct, by showing
a fully-abstract encoding from \HOpi to the $\pi$-calculus. In this section, we show that a 
translation can be developed
also from \CHOP to \CP.
Namely, process terms in \CHOP, which may use process mobility, can be translated to terms in \CP, 
which
may not. We show that the translation supports type preservation and an operational correspondence.

The latest version of the calculus \CP, from~\cite{CLMSW16}, is equivalent to a fragment of \CHOP.
Formally, for this section, let \CP be the calculus obtained from \cref{fig:typing_rules} by
removing the rules \rname{Id}, $\provide{}$, $\assume{}$, and \rname{Chop} (the process mobility 
rules), and
also by removing process contexts ($\Theta$) from all other rules.

We now define a translation $\encho{P}{}$ from processes $P$ in \CHOP to \CP.
The key idea for the translation is that the output of a process--$\sendho x {\abstr \rho P}$---is
translated to sending a reference (a channel) to an instance of $P$, guarded by a series of inputs 
to receive the formal parameters ($\rho$). Dually, the input of a process receives such reference. 
Finally, the invocation of a process variable is translated to a series of outputs which provide the 
actual parameters.
Without loss of generality, we assume that there is an unused set of channels in $P$ indexed by 
process variables, i.e., the set of channels $\{ x^p \mid p \mbox{ is a process variable}\}$. 
Intuitively, $x^p$ is the channel used to encode the behaviour of the process stored in variable 
$p$.

The main rules that define the translation are given in \cref{fig:chop_to_cp_pm}, and follow the
reasoning described above. All the other rules are displayed in \cref{fig:chop_to_cp_other}. In the
rules, we abuse notation and define the translation as taking proof trees in \CHOP to process terms
in \CP. (Similarly to the presentation of the translation of functional programs in the calculus GV 
into \CP by \citet{W14}.)
This eliminates ambiguity on the distribution and usage of names, types, and environments.
We also represent parameter records $\rho$ as ordered according to lexicographic ordering on 
labels.

\begin{figure}
\begin{displaymath}
\begin{array}{r@{}c@{\ }ll}
\enchopt{
	\infer[\rname{Chop}]
	{
		\judge{\Theta,\Theta'}{Q \letp p{\abstr {\rho} P}}{\Gamma}
	}
	{
		\judge{\Theta}{P}{\Delta\rho}
		&
		\judge{\Theta', p:\Delta}{Q}{\Gamma}
	}
}{}
&=&
	\cpres{x^p}{\encho{\Delta}{}}{y^p}{} \left(
		\encho{Q}{}
		\pp
		\recvs {y^p}{z_1}.\cdots.\recvs{y^p}{z_k}.y^p().\encho{P}{}
	\right)
& \mbox{where } \rho = \{ \til{l=z} \}
\\\\
\enchopt{
	\infer[\rname{Id}]
	{
		\judge{p:\Gamma}{\invoke p{\rho}}{\Gamma\rho}
	}{}
}{}
&=&
	\sendfrees {x^p}{z_1}.\cdots.\sendfrees {x^p}{z_k}.\close {x^p}
& \mbox{where } \rho = \{ \til{l=z} \}
\\\\
\enchopt{
	\infer[\provide{}]{
		\judge{\Theta}{\sendho x{\abstr \rho P}}{x:\provide{\Gamma}}
	}{\judge{\Theta}{P}{\Gamma\rho}}
}{}
&=&
	\sendho xy.\left(
		\recvs y{z_1}.\cdots.\recvs y{z_k}.
		\wait y{\encho{P}{}}
	\pp \close x \right)
& \mbox{where } \rho = \{ \til{l=z} \}
\\\\
\enchopt{
	\infer[\assume{}]{
		\judge{\Theta}{\recv xpQ}{\Gamma,x:\assume{\Delta}}
	}{
		\judge{\Theta,p:\Delta}{Q}{\Gamma}
	}
}{}
&=&
\recv x{x^p}{\wait x{\encho{Q}{}}}
\end{array}
\end{displaymath}
\caption{Translating \CHOP to \CP, Rules for Process Mobility.}
\label{fig:chop_to_cp_pm}
\end{figure}

\begin{figure}
\begin{displaymath}
\begin{array}{c}
\enchopt{
	\infer[\rname{Axiom}]
	{\judge{}{\phl{\cpaxiom{x}{A^\perp}{y}{A}}}{\phl x:\thl {A^\bot}, \phl y:\thl{A}}}
	{}
}{}
=
\cpaxiom{x}{}{y}{\encho{A}{}}
\\\\
\enchopt{
	\infer[\rname{Cut}]
	{\judge{\Theta,\Theta'}{\phl{\cpres{x}A{y}{A^\perp}(P \pp Q)}}{\thl\Gamma,\thl\Delta}}
	{\judge{\Theta}{\phl P}{\thl\Gamma, \phl x:\thl A} & \judge{\Theta'}{\phl Q}{\thl\Delta,\phl
y:\thl{\dual{A}}}}
}{}
=
\cpres{x}{\encho{A}{}}{y}{}(\encho{P}{} \pp \encho{Q}{})
\\\\
\enchopt{
	\infer[\tensor]{
		\judge{\Theta,\Theta'}{\send xyPQ}{\Gamma,\Delta, x:A\tensor B}
	}{
		\judge{\Theta}{P}{\Gamma,y:A}
		&
		\judge{\Theta'}{Q}{\Delta,x:B}
	}
}{}
=
\send xy{\encho{P}{}}
{\encho{Q}{}}
\\\\
\enchopt{
	\infer[\parr]{
		\judge{\Theta}{\recv xyP}{\Gamma, x:A \parr B}
	}{
		\judge{\Theta}{P}{\Gamma, y:A, x:B}
	}
}{}
=
\recv xy{\encho{P}{}}
\\\\
\enchopt{
	\infer[\oplus_1]
	{\judge{\Theta}{\phl{\inl xP}}{\thl\Gamma, \phl x:\thl{A \oplus B}}}
	{\judge{\Theta}{\phl P}{\thl\Gamma, \phl x:\thl A}}
}{}
=
\inl x{\encho{P}{}}
\quad
\enchopt{
	\infer[\oplus_2]
	{\judge{\Theta}{\phl{\inr xP}}{\thl\Gamma, \phl x:\thl{A \oplus B}}}
	{\judge{\Theta}{\phl P}{\thl\Gamma, \phl x:\thl B}}
}{}
=
\inr x{\encho{P}{}}
\\\\
\enchopt{
	\infer[\with]
	{\judge{\Theta}{\phl{\piCase xPQ}}{\thl \Gamma, \phl x:\thl{A \with B}}}
	{\judge{\Theta}{\phl P}{\thl\Gamma, \phl x:\thl A} &
		\judge{\Theta}{\phl Q}
		{\thl\Gamma, \phl
		x:\thl B}}
}{}
=
\piCase x{\encho P{}}{\encho Q{}}
\\\\
\enchopt{
	\infer[?]
	{\judge{\Theta}{\phl{\client xyP}}{\thl \Gamma, \phl x:\thl {\query A}}}
	{\judge{\Theta}{\phl P}{\thl \Gamma, \phl y:\thl A}}
}{}
=
\client xy{\encho P{}}
\quad
\enchopt{
	\infer[!]
	{\judge{}{\phl{\srv xyP}}{\thl{\query \Gamma}, \phl x:\thl{\bang A}}}
	{\judge{}{\phl P}{\thl{\query \Gamma}, \phl y:\thl A}}
}{}
=
\srv xy{\encho P{}}
\\\\
\enchopt{
	\infer[\forall]
	{\judge{\Theta}{\phl{\recvtype xXP}}{\thl\Gamma, \phl x:\thl{\forall X.B}}}
	{\judge{\Theta}{\phl P}{\thl\Gamma, \phl x:\thl B & \thl X \notin \m{ftv}(\thl \Theta) \cup
		\m{ftv}(\Gamma)}}
}{}
=
\recvtype xX{\encho P{}}
\\\\
\enchopt{
	\infer[\exists]
	{\judge{\Theta}{\phl{\sendtype xAP}}{\thl \Gamma, \phl x:\thl{\exists X.B}}}
	{\judge{\Theta}{\phl P}{\thl \Gamma, \phl x:\thl{B\{A/X\}}}}
}{}
=
\sendtype x{\encho{A}{}}{\encho P{}}
\quad
\enchopt{
	\infer[\rname{Weaken}]
       {\judge{\Theta}{\phl P}{\thl \Gamma, \phl x:\thl{\query A}}}
       {\judge{\Theta}{\phl P}{\thl \Gamma}}
}{}
=
\encho P{}
\\\\
\enchopt{
	\infer[\rname{Contract}]
	{\judge{\Theta}{\phl{P\{x/y, x/z\}}}{\thl\Gamma, \phl x:\thl{\query A}}}
	{\judge{\Theta}{\phl P}{\thl\Gamma, \phl y:\thl{\query A}, \phl
         z:\thl{\query A}}}
}{}
=
\encho{P\{x/y, x/z\}}{}
\\\\
\enchopt{
	\infer[\one]
	{\judge{\Theta}{\phl{\close x}}{\phl x: \thl\one}}
	{ }
}{}
= \close x
\quad
\enchopt{
	\infer[\bot]
	{\judge{\Theta}{\phl {\wait xP}}{\thl \Gamma, \phl x:\thl\bot}}
	{\judge{\Theta}{\phl P}{\thl\Gamma}}
}{}
= \wait x{\encho P{}}
\\\\
\enchopt{
	\infer[\top]
	{\judge{\Theta}{\phl{\emptyCase x}}{\thl{\Gamma}, \phl x:\thl\top}}
	{ }
}{}
= \emptyCase x
\end{array}
\end{displaymath}
\caption{Translating \CHOP to \CP, Other Rules.}
\label{fig:chop_to_cp_other}
\end{figure}

The translation makes use of a type translation $\encho{A}{}$ from types in \CHOP to types in \CP,
since \CP does not have the types $\provide{\Delta}$ and $\assume{\Delta}$. The type translation 
is defined in
\cref{fig:type_translation}.

\begin{figure}
\begin{displaymath}
\begin{array}{rcl}
\encho{\provide{\Delta}}{} &=& \dual{\left(\encho{\Delta}{}\right)} \tensor \one
\\
\encho{\assume{\Delta}}{} &=& \encho{\Delta}{} \parr \bot
\\
\encho{(l_i:A_i)_i}{} &=& 
\encho{\dual{A_1}}{} \tensor \cdots
\tensor
\encho{\dual{A_n}}{}
\\\\
\multicolumn{3}{l}{
	\encho{A\tensor B}{} = \encho{A}{} \tensor \encho{B}{}
	\quad
	\encho{A \parr B}{} = \encho{A}{} \parr \encho{B}{}
}
\\
\multicolumn{3}{l}{
	\encho{A \oplus B}{} = \encho{A}{} \oplus \encho{B}{}
	\quad
	\encho{A \with B}{} = \encho{A}{} \with \encho{B}{}
}
\\
\multicolumn{3}{l}{
	\encho{\zero}{} = \zero
	\quad
	\encho{\top}{} = \top
	\quad
	\encho{\one}{} = \one
	\quad
	\encho{\bot}{} = \bot
	\quad
	\encho{X}{} = X
}
\\
\multicolumn{3}{l}{
	\encho{\query A}{} = \query \encho{A}{}
	\quad
	\encho{\bang A}{} = \bang \encho{A}{}
	\quad
	\encho{\exists X.A}{} = \exists X. \encho{A}{}
	\quad
	\encho{\forall X.A}{} = \forall X. \encho{A}{}
}
\end{array}
\end{displaymath}
\caption{Translation of types from \CHOP to \CP.}
\label{fig:type_translation}
\end{figure}

The translation $\encho{P}{}$ preserves types, up to usages of process variables, which are
translated
to usages of the corresponding channels that implement these variables. To state this formally, we
need
to translate types of process variables into channel types as follows.
\begin{displaymath}
\begin{array}{rcl}
\encho{\cdot}{} &=& \cdot
\\
\encho{\Theta,p:\Delta}{} &=& \encho{\Theta}{}, x^p:\encho{\Delta}{}
\end{array}
\end{displaymath}

\begin{theorem}
\label{thm:translation_type_preservation}
Let $\judge{\Theta}{P}{\Gamma}$ in \CHOP. Then, $\judge{}{\encho{P}{}}{\encho{\Theta}{},\encho{\Gamma}{}}$
in
\CP.
\end{theorem}
\begin{proof}
By induction on the structure of $P$. For the proof cases, just apply the typing rules of
\CP by following the structure of the process terms on the right of the translation rules.
\end{proof}

Notice how the process variables that a process depends on in \CHOP become extra channels that the
translation needs in \CP ($\encho{\Theta}{}$), as expected. Of course, if the original \CHOP process does not
depend
on any free process variables, then the result of the translation does not depend on any such
extra channels, as formalised by the following corollary.

\begin{corollary}
\label{cor:closed_translation}
Let $\judge{\ }{P}{\Gamma}$ in \CHOP. Then, $\judge{\ }{\encho{P}{}}{\encho{\Gamma}{}}$ in \CP.
\end{corollary}

\newcommand{\mto}{\reducesto^{*}}
\newcommand{\pto}{\reducesto^{+}}

Let $\mto$ be the transitive closure of $\reducesto$ and $\pto$ be defined as $\mto\reducesto$. The translation supports the following operational correspondence.

\begin{theorem}
$\judge{\Theta}{P}{\Gamma}$ implies the following properties.
\begin{itemize}
\item (Completeness) If $P \to P'$, then $\encho{P}{} \pto \encho{P'}{}$.
\item (Soundness) If $\encho{P}{} \mto Q$, then $P \pto P'$ for some $P'$ such that $Q \mto \encho{P'}{}$.
\end{itemize}
\end{theorem}

\section{Related Work}
\label{sec:related}

Since its inception, linear logic was described as the logic of concurrency \citep{G87}.
Correspondences between the proof theory of linear logic and variants of the $\pi$-calculus emerged soon
afterwards \citep{A94,BS94}, by interpreting linear logic propositions as types for channels. Later,
linearity inspired also the seminal theories of linear types for the $\pi$-calculus \cite{KPT99}
and session types \cite{HVK98}. Even though the two theories do not use exactly linear logic, the work by \citet{DGS17} shows that the link is still strong enough that session types can be encoded into linear types.

It took more than ten years for a formal correspondence between linear logic and (a variant of) session
types to emerge, with the seminal paper by \citet{CP10}. This then inspired the development of Classical Processes
(\CP) by \citet{W14},
which we have already discussed plentifully in this article. We have extended this line of work
to include process mobility. Process mobility has been studied deeply in the
context of the $\pi$-calculus, starting from the inception of the Higher-Order $\pi$-calculus
(\HOpi) by \citet{S93}. The definition of \HOpi does not require a typing discipline. Differently, 
the definition of \CHOP is based on a typing discipline that treats process variables 
linearly. Also, our semantics is not
defined a-priori as in most calculi, but rather derived by sound proof transformations allowed by our proof theory.
The proof theory of \CHOP generalises linear logic by allowing to assume that some judgements
can be proven, and to provide evidence for resolving these assumptions.
Similar ideas have been used in the past in different contexts, for example for modal
logic~\cite{NPP08} and logical frameworks~\cite{BS15}.

The concept of explicit substitution has been originally introduced to formalise
execution strategies for the $\lambda$-calculus that are more amenable to efficient
implementations \cite{ACCL91}. The proof theory of \CHOP naturally yields a theory of explicit 
substitutions for higher-order processes.
\cite{S93} used a similar syntax as syntactic sugar; its desugaring
is similar to that of our procedures in \cref{sec:extensions}.

The first version of multiparty session types is discussed in detail in~\cite{HYC16}.
Multiparty session types is currently a popular research topic,
and there is a substantial body of work on multiparty sessions. Our multiparty
version of \CHOP in \cref{sec:mchop} is the first that introduces process mobility to this line of
work, and illustrates how typing can be used to guarantee progress in this setting.
The first formulation of how multiparty sessions may be supported in the setting
of linear logic (without process mobility) was given in~\cite{CMSY15}, and later investigated
further in~\cite{CMSY17} and in~\cite{CLMSW16}.

Other session calculi include primitives for moving processes by relying on a functional
layer~\cite{TCP13,MY15}. Differently, \CHOP offers the first logical reconstruction of (a linear 
variant of) \HOpi~\cite{S93}, where mobile code is just processes, instead of functions (or values 
as intended in $\lambda$-calculus). The key difference with \citep{TCP13} is that \CHOP treats 
process variables linearly.
Of course, process abstractions and functions can be seen as equivalent ideas.
We chose the abstraction formulation because it allows us to use \CLL contexts as higher-order 
types directly, making the theory of \CHOP simpler.
For example, we do not require the additional asymmetric connectives in session types used
in~\cite{TCP13} for communicating processes ($\tau \supset A$ and $\tau \wedge A$).
The ``send a process and continue over channel $x$'' primitive found in~\cite{TCP13,MY15} can be
encoded in \CHOP, as shown in \cref{sec:proc_output_wc}. Since process variables are non-linear functional values in \cite{TCP13}, the usage of process variables does not follow a discipline of linearity as in \CHOP, offering less control. However, the functional layer in~\cite{TCP13}
allows for a remarkably elegant integration of recursive types, which we left out of the scope of this article. A potential direction to recover this feature is the work presented in~\cite{TCP14}. 
Our notion of explicit substitution for higher-order processes is new, and may
be adopted also in the settings of~\cite{TCP13} and~\cite{MY15} for devising efficient execution
strategies. The work in~\cite{TCP13} also has a logical basis, like ours. \CHOP has the advantage
of being formulated as an extension of classical linear logic, in a way that integrates well with
existing features. This allows us, for example, to achieve multiparty sessions, which at this time 
is still unclear how to do in the intuitionistic setting of~\cite{TCP13}. No encoding of 
higher-order processes to first-order processes is provided in~\cite{TCP13,MY15}.

The calculus of Linear Compositional Choreographies (LCC) \cite{CMS17}
gives a propositions as types correspondence for Choreographic Programming~\cite{M13:phd} based on 
linear logic. \CHOP may provide the basis for extending LCC with process mobility,
potentially yielding the first higher-order choreography calculus.

\citet{A17} investigated a notion of observational equivalence for \CP and a denotational semantics 
to capture it. His formulation requires the introduction of additional syntax and typing rules, in 
order to define parallel compositions that leave the names of connected channels free and include 
processes with empty behaviour ($\rname{Mix}_0$). 
Thus, we left an investigation of observational equivalence for \CHOP to future work, which would 
require extending Atkey's work to our higher-order constructs. We conjecture that extending the 
denotational semantics defined by Atkey can benefit from our translation of higher-order types as 
channel types. An immediate application would be to prove a full abstraction result for our 
translation from \CHOP to \CP---this would not be very surprising, given the structure of our 
translation, and also because the denotational semantics of higher-order types would probably be 
defined similarly to our type translation.

\section{Conclusions}
\label{sec:conclusions}

We presented the calculus of Classical Higher-Order Processes (\CHOP), an extension of Classical
Processes (\CP) \cite{W14} to mobile code.
Our formulation extends Classical Linear Logic (\CLL) to higher-order reasoning, viewing proofs as linear ``resources'' that can be used to assume premises in other proofs.
\CHOP integrates well with the existing features of linear logic, as we illustrated in 
\cref{sec:chop,sec:extensions}.

The translation from \CHOP to \CP shows how our higher-order processes that use mobile 
code can be simulated by using reference passing instead.
This result is distinct from the original one by \citet{S93}, because we are operating in a typed 
setting.
Understanding what a calculus with behavioural types, such as session types, can express---i.e., 
what well-typed terms can model---is a nontrivial challenge in general~\cite{P16}.
Our translation shows that the proof theory of \CP is powerful enough to simulate higher-order 
processes in the proof theory of \CHOP.
In practice, this has the usual implications: the translation gives us the possibility to write 
programs that use code mobility andthen choose later whether we should really use code mobility or 
translate it to an implementation basedon reference passing.
This choice depends on the application case. If we are modelling the transmission of an
application to be run somewhere else (as in cloud computing), then code mobility is necessary.
Otherwise, if we are in a situation where we can choose freely, then we should choose whichever 
implementation is more efficient. For example, code mobility is useful if two processes, say a 
client and a server, are operating on a slow connection; then, instead of performing many 
communications over the slow connection, the client may send an application to the server such that 
the server can communicate with
the application locally, and then send to the client only the final result.
Lastly, if we are using code mobility in an environment where communications are implemented in
local memory (as in many object-oriented language implementations or other emerging languages, like
Go), then the translation gives us a compilation technique towards a
simpler language without code mobility (\CP), which we can use to simplify runtime implementations.

Process mobility is the underlying concept behind the emerging interest on runtime adaptation, a mechanism by which processes can receive updates to their internal code at runtime.
Different attempts at formalising programming disciplines for runtime adaptations have been
made, e.g., by \citet{DP15,DGGLM16,BCDV17}, but none are rooted in a propositions as types correspondence and
all offer different features and properties.
\CHOP brings us one step nearer to formulating adaptable processes based on the firm foundations of linear logic.

\begin{acks}
The author thanks Lu\'\i s Cruz-Filipe, Marco Peressotti, and Davide Sangiorgi for useful 
discussions.
\end{acks}

\bibliographystyle{ACM-Reference-Format}
\bibliography{biblio}

\end{document}